\documentclass[10pt,journal,twocolumn]{IEEEtran}
\usepackage{amsthm}
\usepackage{amssymb}
\usepackage{enumerate}
\usepackage{mathrsfs}
\usepackage{epsfig}
\usepackage{yhmath}
\def\mf{\mathfrak}

\def\m{\mathcal}
\def\md{\mathbb}
\def\ms{\mathscr}
\def\eps{\varepsilon}
\def\tn{\textnormal}
\def\wt{\widetilde}
\def\RealF{\md{R}}
\def\NaturalF{\md{N}}
\def\IntF{\md{Z}}
\def\Expt{\md{E}}
\def\Prob{\md{P}}
\newcommand{\frc}[1]{\langle #1 \rangle}
\newcommand{\uico}{[\hspace{0.08em}0,\hspace{-0.05em} 1\hspace{-0.08em})}
\newcommand{\dfn}{ \stackrel{\tn{def}}{=} }

\def\argmin{\mathop{\rm argmin}}

\def\pmin{p_{\rm min}}
\def\pmax{p_{\rm max}}
\def\qmax{q_{\rm max}}

\newtheorem{lemma}{Lemma}
\newtheorem{corollary}{Corollary}
\newtheorem{theorem}{Theorem}

\newtheorem{remark}{Remark}

\newcommand{\bi}[1]{\left[{#1}\right)}
\newcommand{\mbi}[1]{{\rm bin}\hspace{-0.2em}\left({#1}\right)}

\allowdisplaybreaks[3]

\begin{document}

\title{Delay and Redundancy in Lossless Source Coding}

\author{Ofer~Shayevitz, Eado~Meron, Meir~Feder and Ram~Zamir
\thanks{The authors are with the Department of EE-Systems, Tel Aviv University, Tel Aviv, Israel \{email: ofersha@eng.tau.ac.il, meroneado@gmail.com, meir@eng.tau.ac.il, zamir@eng.tau.ac.il\}. This paper was presented in part at ISIT 2006, DCC 2007 and DCC 2008. The work of O. Shayevitz was partially supported by the Adams Fellowship and the ITA fellowship. The work of R. Zamir was partially supported by the Israel Academy of Science, ISF grant number 870/11.}}

\maketitle


\begin{abstract}
The penalty incurred by imposing a finite delay constraint in lossless source coding of a memoryless source is investigated. It is well known that for the so-called block-to-variable and variable-to-variable codes, the redundancy decays \textit{at best polynomially} with the delay, where in this case the delay is identified with the source block length or maximal source phrase length, respectively. In stark contrast, it is shown that for sequential codes (e.g., a delay-limited arithmetic code) the redundancy can be made to decay \textit{exponentially} with the delay constraint. The corresponding redundancy--delay exponent is shown to be at least as good as the R\'enyi entropy of order $2$ of the source, but (for almost all sources) not better than a quantity depending on the minimal source symbol probability and the alphabet size.
\end{abstract}

\thispagestyle{empty}

\section{Introduction}
It is well known that any memoryless source can be asymptotically losslessly compressed to its entropy \cite{cover}. However, in the presence of resource constraints, a rate penalty, referred to as \textit{redundancy}, is unavoidable. In this work we focus on the redundancy in the encoding of a  memoryless source incurred by the imposition of a \textit{strict end-to-end delay constraint} $d$ measured in source clocks, i.e., under the requirement that the $n$-th encoded symbol must always be perfectly reproduced at the decoder by time $n+d$.

In the lossless source coding literature, three classes of codes in which delay is a design parameter are traditionally studied: 1) The Block-to-Variable (BV) class (e.g. Huffman code \cite{Huffman52}), where a source sequence is partitioned into equi-length blocks and each block is mapped to a unique variable length codeword from a prefix-free set, 2) The Variable-to-Block (VB) class (e.g. Tunstall code \cite{TunstallPhd, Khodak69}), where the source sequence is parsed into phrases according to a complete code-tree, and each phrase is mapped to a unique fixed length codeword, and 3) The Variable-to-Variable (VV) class (e.g., Khodak codes), where the source sequence is similarly parsed but each phrase is mapped to a unique variable length codeword from a prefix-free set. In the sequel, we collectively refer to the three classes above as \textit{the classical framework}. In the BV class, a delay constraint is interpreted as a block length constraint, and the redundancy is known to decay at best polynomially with the delay \cite{Szpankowski2000,Drmota2006}. In the VB/VV class (where the delay is a random variable depending on the source sequence) the delay constraint is translated into a maximal phrase length constraint, and the redundancy again decays at best polynomially with the delay, though sometimes faster than in the BV case \cite{Khodak69, Khodak72,Bugeaud2004}\footnote{These results hold even in the weaker case of an expected delay constraint.}.

In a delay constrained setting, the classical framework admits two (related) limitations. First, even within that framework, there is an apparent disparity between delay and block/phrase length. The reason block/phrase lengths are identified with delay in the first place is since concatenating codewords allows the source reproduction at block/phrase length intervals. However, the delay can sometimes be significantly shorter, for essentially the same reason: Consider a BV code of block length $n=kd$ obtained by concatenating $k$ BV codes of block length $d$. Clearly, the decoder can reproduce symbols with a delay $d$, rather than the possibly much larger delay $n$. Waiting until the end of the block would mean the encoder is ``holding back'' bits it is already certain of, clearly an undesirable trait in a delay constrained setting. Of course, the redundancy associated with such an encoder in the limit of $k\to\infty$ still decays polynomially as a function of $d$, which brings us to the second limitation. In the memoryless classical framework, the encoder never looks beyond the end of the current block/phrase, in the sense that the source's prefix has no effect on the output of the encoder beyond that point\footnote{This assertion does not hold for sources with memory, where dependencies between phrases can be beneficial \cite{Tjalkens-Willems}.}. The encoder is therefore being ``reset'' roughly every $d$ symbols. Loosely speaking, the penalty incurred by forcing these regularly recurring reset points, is the source of the polynomial delay of the redundancy.

With these observations in mind, we recall a lossless coding technique of a different flavor that does not suffer from the above shortcomings. In \textit{arithmetic coding} \cite{abramson_book,arithmetic_coding_jelinek,RissanenLangdon1979,WittenEtAl1987}, a source subsequence is sequentially mapped into nested subintervals of the unit interval, with length equal to the sequence probability, and the common most significant bits of the current subinterval are emitted. This way, the encoder never holds back any bits it is already certain of, by definition. Moreover, whereas BV/VB/VV encoders never look beyond the end of the current block/phrase, an arithmetic encoder constantly looks into the (possibly infinite) future. Unfortunately, this comes at a cost of an unbounded delay (though a bounded expected delay, see \cite{gallager_lecture_notes,arithmetic_coding_savari,ShayevitzDelay2006}). Nevertheless, the notion of arithmetic coding does point us in the right direction. In a delay constrained framework, an encoder should \textit{by definition} be sequential, emitting all the bits it can at any given instance. Moreover, a good delay constrained encoder should always strive to look $d$ steps ahead, avoiding ``reset'' points as much as possible. As we shall see, these properties are nicely captured within an interval mapping type framework.

\nocite{delay_redundancy_DCC2007,delay_redundancy_DCC2008}

In this paper, we introduce a general framework for lossless delay constrained coding of a memoryless source, and study the fundamental tradeoff between delay and redundancy. We show that, in stark contrast to the polynomial decay within the classical framework, the redundancy $\mf{R}(P,d)$ associated with a memoryless source $P$ over a finite alphabet $\m{X}$, can be made to decay \textit{exponentially} with the delay $d$. Specifically, we show that any encoder obeying a delay constraint $d$ satisfies \footnote{By $a_d\lessapprox b_d$ we mean $\liminf_{d\rightarrow\infty}\frac{1}{d}\log \frac{b_d}{a_d} >0$.}
\begin{equation*}
\left(\frac{\pmin}{|\m{X}|}\right)^{8d} \,\lessapprox \,\mf{R}(P,d) \,\lessapprox \,\pmax^d
\end{equation*}
where $\pmin,\pmax$ are the minimal and maximal source symbol probabilities, the upper bound holds for all sources, and the lower bound holds for almost all sources\footnote{Recall that the reason for jointly coding over multiple source symbols, and consequently incurring delay, is to make the rounding error of the log-probabilities negligible. This is unnecessary for dyadic sources, where symbol probabilities are all integer powers of 2. Hence, a lower bound cannot hold for all sources, as dyadic sources can attain zero redundancy with zero delay.}. We then tighten the upper bound and obtain
\begin{equation*}
\mf{R}(P,d) \,\lessapprox \, 2^{-dH_2(P)}
\end{equation*}
where $H_2(P)$ is the R\'enyi entropy of order $2$ of the source. For our upper bound, we introduce a construction based on mismatched arithmetic coding in conjunction with a fictitious symbol insertion mechanism. For our lower bound, we provide a ``generalized interval mapping'' representation for delay constrained encoders. 

\textit{Related work.} Whereas in this paper we consider the impact of an end-to-end delay constraint measured in source clocks, other works have considered complementary questions where delay is measured in encoded bits. In \cite{arithmetic_coding_savari,SavariReport} the authors describe a variable-length lossless source coding system based on finite precision arithmetic coding, that falls outside the classical framework and is of a similar flavor to the codes considered herein; Specifically, they show \cite[Appendix II]{SavariReport} that the associated redundancy decays exponentially with the maximal number of encoded bits the decoder can hold in its queue. A similar observation can be deduced from the discussion in \cite{MoffatEtAl1998}. While employing a different measure of delay, it appears plausible (but remains unverified) that these constructions could also be employed to derive an exponential upper bound on the redundancy as a function of the delay in source clocks. None of these prior works provided a lower bound for the redundancy. In \cite{jelinek68}, the author considers a setting where the channel connecting the encoder and the decoder can transmit a fixed number of bits per second, and has a finite length queue at its input. He shows that the probability of queue overflow for BV codes can be made to decay exponentially with the size of the queue, and describes the tradeoff between the exponent and the minimal achievable compression rate. 

\textit{Organization}. Our framework is introduced in Section \ref{sec:perlim}, and some basic lemmas are derived. In Section \ref{sec:delay_AC}, the delay profile of mismatched arithmetic coding is analyzed. This analysis is then applied in Section \ref{sec:lower_bound} where a lower bound on the redundancy--delay exponent is derived. In Section \ref{sec:upper_bound}, a corresponding upper bound on the redundancy--delay exponent for almost all sources is presented. Some final remarks are given in Section \ref{sec:conc}.

\section{Preliminaries}\label{sec:perlim}
\subsection{Notations}
We write $s\preceq t$ to indicate that a string $s$ is a prefix of a string $t$, and $s\prec t$ to indicate that $s\preceq t$ and $s\neq t$. A set of finite strings $S$ is said to be \textit{prefix-free} if no pair of strings $s,t\in S$ satisfies $s\prec t$. The \textit{longest common prefix} of $S$ is the string $t$ of maximal length satisfying $t\preceq s$ for all $s\in S$. The Lebesgue measure of a set $A\subseteq\RealF$ is denoted by $|A|$. The \textit{fractional part} of a number $a\in\RealF$ is denoted by $\frc{a}\dfn a-\lfloor a\rfloor$. The \textit{difference modulo-1} $\frc{A-B}$ between two sets $A,B\subseteq\RealF$ is the set of all numbers $\frc{a-b}$ where $a\in A\,,b\in B$. For any function $f:\RealF\mapsto\RealF$ and any set $A\subseteq\RealF$, we write $f(A)$ for the image of $A$ under $f$. All logarithms are taken to the base of $2$. A \textit{total order} of a finite set is called simply an \textit{order}.

The following lemma is easily verified.
\begin{lemma}\label{lem:diff_set}
Let $A,B\subseteq\RealF$ be any two sets. Then
\begin{enumerate}[(i)]
\item If $b\in B$ and $\frc{c}\not\in\frc{A-B}$, then $b+c\not\in A$.
\item If $b\in B$ and $\frc{\log{c}}\not\in\frc{\log{A}-\log{B}}$, then $bc\not\in A$.
\end{enumerate}
\end{lemma}

\subsection{Sources}
Let $\m{X}$ be a finite alphabet of source symbols. The set of all length-$n$ strings of symbols from $\m{X}$ is denoted $\m{X}^n$, the set of all finite length strings is denoted $\m{X}^*$, and the set of all infinite length strings is denoted $\m{X}^\infty$. We sometimes use the notations $x^n\dfn x_1x_2\ldots x_n$ and $x_m^n\dfn x_mx_{m+1}\ldots x_n$ for finite source strings, where the convention is that $x_m^n=\emptyset$ when $m>n$. A \textit{discrete memoryless source (DMS)} $P$ is defined by a \textit{probability mass function (p.m.f.)} $\{P(x):x\in\m{X}\}$ which naturally induces a product measure over $\m{X}^*$, via $P(st) = P(s)P(t)$ for all $s,t\in\m{X}^*$, where $st$ is the concatenation of $s$ and $t$. Specifically, we denote by $P^n$ the p.m.f. obtained by restricting $P$ to $\m{X}^n$. An infinite random source string emitted by the source $P$ will be denoted by $X^\infty$. The minimal and maximal symbol probabilities under $P$ are denoted $p_{\rm min}$ and $p_{\rm max}$ respectively. The \textit{entropy} of the source is denoted $H(P)$. The \textit{Kullback-Leibler distance}, or \textit{divergence}, between two sources $P,Q$ over the same alphabet is denoted $D(P\|Q)$. We write $P\ll Q$ if $Q(x)=0$ implies $P(x)=0$ for all $x\in\m{X}$. The set of all p.m.f.'s over $\m{X}$ is denoted $\ms{P}(\m{X})$. The \textit{type} of a sequence $x^n\in\m{X}^n$ is the p.m.f. $P_{x^n}\in\ms{P}(\m{X})$ corresponding to the relative frequency of symbols in $x^n$. The set of all possible types of sequences $x^n$ is denoted $\ms{P}^n(\m{X})$. The \textit{type class} of any type $Q\in\ms{P}^n(\m{X})$ is the set $T_Q\dfn \{x^n\in\m{X}^n: P_{x^n}=Q\}$. For $\eps>0$, let $\ms{P}^n_\eps(\m{X},P)\subseteq\ms{P}^n(\m{X})$ be the subset of all types $Q$ for which $\|P-Q\|_\infty<\eps$.

The following facts are well known \cite{csizar_korner}.
\begin{lemma}\label{lem:types}
For any type $Q\in\ms{P}^n(\m{X})$ and any $x^n\in T_Q$:
\begin{enumerate}[(i)]
\item $P(x^n) = 2^{-n(D(Q\|P)+H(Q))}$.
\item $|\ms{P}^n(\m{X})|^{-1}2^{nH(Q)} \leq |T_Q| \leq  2^{nH(Q)}$.
\item $|\ms{P}^n(\m{X})|  = {n+|\m{X}|-1 \choose |\m{X}|-1} \leq (n+1)^{|\m{X}|}$.
\item (AEP) For any $\eps>0$,
\begin{equation*}
\lim_{n\rightarrow\infty}P\left(\bigcup_{Q\in\ms{P}^d_\eps(\m{X},P)}T_Q\right) = 1
\end{equation*}
\end{enumerate}
\end{lemma}

The R\'enyi entropy \cite{Renyi1960} of order $\alpha$ of a source $P$ is
\begin{equation*}
H_\alpha(P) \dfn \frac{1}{1-\alpha} \log\sum_{x\in\m{X}}(P(x))^\alpha
\end{equation*}
\begin{lemma}[From \cite{ShayevitzRenyiISIT}]\label{lem:renyi}
The R\'enyi entropy of order $\alpha>1$ admits the following variational characterization:
\begin{equation*}
H_\alpha(P) = \min_{Q\in\ms{P}(\m{X})}\left\{\frac{\alpha}{\alpha-1}\,D(Q\|P)+H(Q)\right\} \\
\end{equation*}
For $0<\alpha<1$, replace the $\min$ with a $\max$.
\end{lemma}

For any two sources $P,Q$ over the same alphabet $\m{X}$, we define
\begin{equation*}
\nu(P,Q) \dfn \sup_{x\in\m{X}:P(x)>0}\frac{P(x)}{Q(x)}
\end{equation*}
The following is easy to verify.
\begin{lemma}
$1\leq \nu(P,Q)\leq \infty$ with equality in the lower bound if and only if $P=Q$, and in the upper bound if and only if $P \not\ll Q$. 
\end{lemma}

\subsection{Encoders and Decoders}\label{subsec:encoders}
An \textit{encoder} is a mapping $\m{E}:\m{X}^*\mapsto \{0,1\}^*$ such that for any $s\in\m{X}^*$, $\m{E}(s)$ is the longest common prefix of the set of bit strings $\{\m{E}(sx) : x\in\m{X}\}$. Namely, we are assuming the encoder does not withhold any bits; at any given time, the longest prefix the encoder is certain of is assumed to have already been emitted.. This will be referred to as \textit{the integrity property}. Note that the integrity property implies in particular the \textit{consistency property}, namely that $\m{E}(s)\preceq\m{E}(sx)$.

An encoder $\m{E}$ induces a \textit{decoder}, which is a partial mapping $\m{D}_\m{E}:\{0,1\}^*\mapsto \m{X}^*$, defined as follows. For any $b\in\{0,1\}^*$, let
\begin{equation*}
\m{E}^{-1}(b)\dfn \{s\in\m{X}^* : b \preceq \m{E}(s) \}
\end{equation*}
Then $\m{D}_\m{E}(b)$ is the longest common prefix of $\m{E}^{-1}(b)$ if the latter set is not empty, and is otherwise undefined. Note that by definition, $\m{D}_\m{E}$ does not withhold any symbols, hence satisfies a similar integrity property. Furthermore, $\m{D}_\m{E}$ is defined not only over the range of $\m{E}$, but also on the set of all prefixes thereof; the decoder hence operates without the need to be synced with the source clock. Since a decoder is uniquely defined by an encoder, we shall focus our discussion hereafter on encoders only.

An encoder $\m{E}$ is associated with a \textit{delay function}, which returns the minimal number of symbols from a given (infinite) suffix that needs to be encoded so that a given prefix can be fully decoded. Formally, the delay function is a mapping $\delta^\m{E}:\m{X}^*\times\m{X}^\infty\mapsto \NaturalF\cup\{0,\infty\}$, where $\delta^\m{E}(s,x^\infty)$ is the minimal $k\in\NaturalF\cup\{0\}$ such that $s\preceq \m{D}_\m{E}(\m{E}(sx^k))$. 
If no such $k$ exists, then $\delta^\m{E}(s,x^\infty) \dfn \infty$.

The \textit{delay profile} associated with an encoder $\m{E}$ and a source $P$ for a given prefix $s$, is the following extended-real-valued r.v.:
\begin{equation*}
\Delta^\m{E}(s,P)\dfn\delta^\m{E}(s,X^\infty)
\end{equation*}
The delay profile associated with an encoder $\m{E}$ and a source $P$ is then defined to be
\begin{equation*}
\Delta^\m{E}(P) \dfn \sup_{s\in\m{X}^*}\Delta^\m{E}(s,P)
\end{equation*}

Next, we define several families of encoders.
\subsubsection{Lossless Encoders}
An encoder is said to be \textit{lossless w.r.t. $P$} (where $P$ is omitted when there is no confusion), if
\begin{equation*}
\Prob(\Delta^\m{E}(P) < \infty) = 1,
\end{equation*}
The family of all encoders that are lossless w.r.t. $P$ is denoted $\mf{L}(P)$.

\subsubsection{Bounded Expected Delay Encoders}
An encoder is said to admit a \textit{bounded expected delay w.r.t. $P$} (where $P$ is omitted when there is no confusion), if
\begin{equation*}
\Expt(\Delta^\m{E}(P)) <\infty
\end{equation*}
The family of all encoders with bounded expected delay w.r.t. $P$ is denoted $\mf{B}(P)$. Clearly, $\mf{B}(P)\subset \mf{L}(P)$.

\subsubsection{Delay Constrained Encoders}
An encoder is said to be \textit{delay-constrained}, if
\begin{equation}
\sup_{s\in\m{X}^*, t\in\m{X}^\infty} \delta^\m{E}(s,t) <\infty
\end{equation}
More specifically, such an encoder is also said to be \textit{$d$-delay-constrained}, if the supremum above equals $d$. The family of $d$-constrained encoders is denoted by $\mf{C}_d$.\footnote{Note that growing dictionary encoders such as the Ziv-Lempel encoder \cite{lempel_ziv} do not belong to this family, as their delay grows unbounded. } Clearly, $\mf{C}_d \subset \mf{B}(P)$ for any source $P$.

\subsubsection{Phrase/Block Constrained Encoders}
An encoder $\m{E}$ is said to be \textit{phrase-constrained} if $\m{E}\in\mf{C}_d$ for some $d$, and for any $x^\infty\in\m{X}^\infty$ there exists an index sequence $\{i_k\in\NaturalF\}_{k=1}^\infty$ such that $0<i_{k+1}-i_k\leq d+1$ and
\begin{equation}
\delta^\m{E}(x^{i_k},x_{i_k+1}^\infty)  = 0
\end{equation}
In this case we also say the encoder is \textit{$d$-phrase-constrained}. In the special case where $i_k = (d+1)k$ for all $x^\infty\in\m{X}^\infty$, we say the encoder is \textit{$d$-block-constrained}. The family of all $d$-phrase-constrained (resp. $d$-block-constrained) encoders is denoted by $\mf{C}_d^{\rm phrase}$ (resp. $\mf{C}_d^{\rm block}$). Clearly, $\mf{C}_d^{\rm block} \subset \mf{C}_d^{\rm phrase} \subset \mf{C}_d$.

\begin{remark}
Any encoder $\m{E}\in\mf{C}_d^{\rm block}$ (resp. $\m{E}\in\mf{C}_d^{\rm phrase}$) can generally be written as a prefix-dependent concatenation of BV (resp. VB/VV) codes each with block length (resp. maximal phrase length) at most $d+1$. By prefix-dependent here we mean that the code used in the next block (resp. phrase) can generally depend on the source sequence encoded thus far. Note however that for block (resp. phrase) constrained encoders operating over memoryless sources there is no redundancy gain to be reaped by using prefix-dependency, since the entire prefix can already be decoded and hence is irrelevant (in terms of average code-length) to the encoding of the next block (resp. phrase). Hence for memoryless sources, as far as the redundancy--delay tradeoff is concerned, there is no loss of generality in restricting our attention to concatenations of a single fixed BV (resp. VB/VV) code.

Conversely, any BV (resp. VB/VV) code with block length (resp. maximal phrase length) $k$, adapted to process infinite source strings via concatenation, is a $d$-block-constrained (resp. $d$-phrase-constrained) code for some $d\leq k$. Due to the integrity property requirement, it is generally possible that $d<k$, as the base code itself may be a concatenation of shorter codes. This is however clearly redundant, and without loss of generality we can restrict our attention to \textit{minimal} BV (resp. VB/VV) codes, i.e., codes for which $k=d$.
\end{remark}

\begin{remark}
Following the previous remark, it is worth mentioning an interesting class of codes known as \textit{plurally parsable (PP) codes} \cite{SavariRenewal2000}, which are a generalization of VB/VV codes. In a nutshell, a PP encoder is defined via a finite phrase dictionary $\mf{D}\subset\{0,1\}^*$ and a parsing rule. The dictionary is not a complete code-tree, and hence can induce more than one  parsing for some source sequences; in such cases the parsing rule is employed to determine which of the possible parsings will be used. Typically, a greedy parsing rule is employed, looking for the longest match in $\mf{D}$. It is interesting to note that while clearly any PP code is delay-constrained, any nontrivial PP code, i.e., one that cannot be essentially translated into a (uniquely parsable) VB/VV code\footnote{For example, the PP code given by the incomplete code-tree $\mf{D}=\{0,00,1\}$ together with the greedy parsing rule, can essentially be thought of as a uniquely parsable code given by the complete code-tree $\mf{D}=\{00,01,1\}$, in the sense that the parsing induced by the former is a refinement of the parsing induced by the latter.}, is not block/phrase constrained, as there are source sequences for which the delay is always positive. For example, using the PP code given by the incomplete code-tree $\mf{D}=\{0,000,1,111\}$ together with the greedy parsing rule, the delay incurred for the source sequence $001100110011...$ is always at least $1$. Such PP codes hence always look beyond the end of the current phrase. 
\end{remark}

\subsubsection{Interval--Mapping Encoders}
A binary string $b^k\in\{0,1\}^k$ is said to \textit{represent} a \textit{binary interval}
\begin{equation*}
\bi{b^k}\dfn \left[0.b_1b_2,\ldots b_k0,\; 0.b_1b_2,\ldots b_k1 \right)\subseteq\uico
\end{equation*}
For any set $A\subset\uico$ we write $\mbi{A}$ to denote the minimal binary interval containing $A$, i.e.,
\begin{equation*}
\mbi{A}\dfn \bigcap_{b\in\{0,1\}^*:A\subseteq \bi{b}} \bi{b}
\end{equation*}
The following lemma is easily observed.
\begin{lemma}
For any $b,c\in\{0,1\}^*$,
\begin{enumerate}[(i)]
\item $b\preceq c \; \Leftrightarrow \; \bi{c}\subseteq \bi{b}$.
\item $b\not\preceq c$ and $c\not\preceq b \; \Leftrightarrow \; \bi{b}\cap \bi{c} = \emptyset$.
\end{enumerate}
\end{lemma}

Let $\mf{S} \dfn \left\{[\,a,b)\,|\,0\leq a<b\leq 1\right\}$. An encoder $\m{E}$ is said to be an \textit{interval--mapping encoder}, if there exists a mapping $\m{I}^\m{E}:\m{X}^*\mapsto\mf{S}$, i.e., a mapping of finite source sequences into subintervals of the unit interval, such that the following properties are satisfied:
\begin{enumerate}[(i)]
\item \textit{Minimality}: $\bi{\m{E}(s)} = \mbi{\m{I}^\m{E}(s)}$ for any $s\in\m{X}^*$.

\item \label{prop:disjoint_nesting}\textit{Disjoint nesting}: For all $s\in\m{X}^*$ and all distinct $x,y\in\m{X}$,
\begin{equation*}
\m{I}^\m{E}(sx)\subseteq\m{I}^\m{E}(s), \quad \m{I}^\m{E}(sx)\cap \m{I}^\m{E}(sy)=\emptyset
\end{equation*}
\end{enumerate}
The minimality property means that an interval--mapping encoder emits the bit sequence representing the minimal binary interval containing the interval $\m{I}^\m{E}(s)$. It is easily observed that the minimality and disjoint nesting properties together imply the integrity property. The family of interval mapping encoders is denoted by $\mf{I}$.

Let $<$ be any order of $\m{X}$. A special case of an interval--mapping encoder is an \textit{arithmetic encoder w.r.t. the order $<$ matched to a source $P$}, which is defined as follows:
\begin{eqnarray*}
f_1(x) &\dfn& \sum_{y<x}P(y)\\
f_n(x^n) &\dfn& f_{n-1}(x^{n-1}) + f_1(x_n)P(x^{n-1})
\\
\m{I}^\m{E}(x^n) &\dfn& \left[f_n(x^n),f_n(x^n)+P(x^n)\right)
\label{eq:encode}
\end{eqnarray*}
We omit the reference to a specific order $<$ when there is no confusion, or when the statement holds for any order.

\subsubsection{Generalized Interval--Mapping Encoders}
Let $\mf{S}^*$ be the set of all finite disjoint unions of subintervals from $\mf{S}$. An encoder $\m{E}$ is said to be a \textit{generalized interval--mapping encoder} if there exists a mapping $\m{I}^\m{E}:\m{X}^*\mapsto\mf{S}^*$ satisfying the minimality and disjoint nesting properties above. The family of generalized interval--mapping encoders is denoted by $\mf{I}^*$. Clearly, $\mf{I}\subset\mf{I}^*$.

The following lemma shows that any $d$-delay-constrained encoder admits a generalized interval--mapping representation.
\begin{lemma}\label{lem:gim_rep}
Let $\m{E}\in\mf{C}_d$. Then $\m{E}$ can be represented as a generalized interval--mapping encoder with
\begin{equation}\label{eq:int_map_for_d}
\m{I}^{\m{E}}(s)  = \bigcup_{x^d\in\m{X}^d} \bi{\m{E}(sx^d)}
\end{equation}
Hence, $\mf{C}_d\subset \mf{I}^*$.
\end{lemma}
\begin{proof}
See the Appendix.
\end{proof}
\begin{remark}
The representation in (\ref{eq:int_map_for_d}) is a finite union of (possibly overlapping) binary intervals. It is worth noting that an arithmetic encoder matched to a source cannot generally be written that way, as some of its intervals may only be written as an infinite union of binary intervals. This sits well with the fact that generally, an (idealized) arithmetic encoder has an unbounded delay.
\end{remark}

\subsection{Redundancy}
The \textit{(per symbol) expected codelength} at time $n$ associated with an encoder $\m{E}$ and a memoryless source $P$ is
\begin{equation}
\bar{L}_n^\m{E}(P) \dfn n^{-1}\Expt|\m{E}(X^n)|
\end{equation}
where $X^n\sim P^n$. The \textit{(per symbol) expected redundancy} at time $n$  associated with an encoder $\m{E}$ and a memoryless source $P$ is the gap between the expected codelength and the entropy after $n$ symbols have been encoded, i.e.,
\begin{equation*}
\mf{R}_n^\m{E}(P) \dfn \bar{L}_n^\m{E}-H(P)
\end{equation*}
The corresponding \textit{sup--redundancy} and \textit{inf--redundancy} are defined as
\begin{equation*}
\overline{\mf{R}}^\m{E}(P) \dfn \limsup_{n\rightarrow\infty}\mf{R}^\m{E}_n(P)\,,\qquad \underline{\mf{R}}^\m{E}(P) \dfn \liminf_{n\rightarrow\infty}\mf{R}^\m{E}_n(P)
\end{equation*}

Let us define some useful quantities pertaining to generalized interval--mapping encoders, that will enable us to bound their redundancy in relatively simple terms. A generalized interval--mapping encoder $\m{E}$ induces a measure over $\m{X}^n$, defined by
\begin{equation*}
\mu^\m{E}_n(x^n) \dfn |\m{I}^\m{E}(x^n)|
\end{equation*}
and a conditional induced measure, defined as
\begin{equation*}
\mu^\m{E}_k(x^k|x^n) \dfn \frac{\mu^\m{E}_{n+k}(x^nx^k)}{\mu^\m{E}_n(x^n)}
\end{equation*}
Define:
\begin{equation*}
R^\m{E}_n(P) \dfn \frac{1}{n}D\left(P^n\|\mu^\m{E}_n\right)
\end{equation*}
and let
\begin{equation*}
r_d(x^n) = D\left(P^d\|\mu^\m{E}_d(\cdot|x^n)\right)
\end{equation*}
be the \textit{$d$-instantaneous redundancy}.
\begin{remark}
Note that $\mu^\m{E}_n$ and $\mu^\m{E}_k(\cdot|x^n)$ are not necessarily probability distributions, as they may sum to less than unity. However, clearly it still holds that  $R_n^\m{E}(P)\geq 0, r_d(x^n) \geq 0$.
\end{remark}

The next lemma relates the interval-based notions of redundancy defined above, to the actual operational definition of redundancy of the associated generalized interval--mapping encoders. This correspondence will allow us to think of intervals instead of bits, and will play a central role in the sequel.
\begin{lemma}\label{lem:redundancy_bounds}
The following relations hold:
\begin{enumerate}[(i)]
\item For any $\m{E}\in\mf{I}^*$,
\begin{equation*}
\mf{R}^{\m{E}}_n(P) \leq R_n^\m{E}(P)
\end{equation*}
\item For any $\m{E}\in\mf{C}_d$, there exists a generalized interval--mapping representation $\m{I}^\m{E}$ (e.g., the one in Lemma \ref{lem:gim_rep}) such that
\begin{align*}
&\mf{R}^{\m{E}}_n(P)\geq \left(\frac{n+d}{n}\right)R_{n+d}^\m{E}(P) + \frac{d}{n} H(P)
\\
&\underline{\mf{R}}^\m{E}(P) = \liminf_{n\rightarrow\infty}\frac{1}{nd}\sum_{k=1}^{n}\Expt(r_d(X^k))
\end{align*}
\end{enumerate}
\end{lemma}
\begin{proof}
See the Appendix.
\end{proof}

One would naturally be interested in the redundancy performance that can be guaranteed by employing encoders of different classes. In general, the expected redundancy $\mf{R}_n^\m{E}$ of an encoder $\m{E}$ can be negative for some, or even all $n$. However, the sup and inf--redundancy are nonnegative for all lossless encoders, and bounds in the $d$-block/phrase constrained cases are known.
\begin{lemma}\label{lem:redundancy_bounds_known}
The following statements hold\footnote{Recall that $f(d)=O(g(d)) \;\Rightarrow\;\limsup_{d\rightarrow\infty}\left|\frac{f(d)}{g(d)}\right|<\infty$, and $f(d)=\Omega(g(d))\;\Rightarrow\;\liminf_{d\rightarrow\infty}\left|\frac{f(d)}{g(d)}\right|>0 $}:
\begin{enumerate}[(i)]
\item For any source $P$
    \begin{align*}
    \inf_{\m{E}\in\mf{L}(P)}\overline{\mf{R}}^\m{E}(P)  &= \inf_{\m{E}\in\mf{B}(P)}\overline{\mf{R}}^\m{E}(P)= \inf_{\m{E}\in\mf{L}(P)}\underline{\mf{R}}^\m{E}(P)
    \\
    &= \inf_{\m{E}\in\mf{B}(P)}\underline{\mf{R}}^\m{E}(P)
    = 0
    \end{align*}

\item (\textit{From \cite{cover,Khodak72,Drmota2006}}) For any source
    \begin{equation*}
    \inf_{\m{E}\in\mf{C}_d^{\rm block}}\overline{\mf{R}}^\m{E}(P)  = O(d^{-1})
    \,,\;\;
    \inf_{\m{E}\in\mf{C}_d^{\rm phrase}}\overline{\mf{R}}^\m{E}(P)  = O(d^{-\frac{5}{3}})
    \end{equation*}

\item (\textit{From \cite{Szpankowski2000,Drmota2006}}) For almost all sources,
    \begin{align*}
    \inf_{\m{E}\in\mf{C}_d^{\rm block}}\underline{\mf{R}}^\m{E}(P)  &= \Omega(d^{-1})
    \\
    \inf_{\m{E}\in\mf{C}_d^{\rm phrase}}\underline{\mf{R}}^\m{E}(P)  &= \Omega(d^{-2|\m{X}|-1-\eps})
\end{align*}
where $\eps>0$.
\end{enumerate}
\end{lemma}

We see that employing block/phrase-constrained codes for compression under a strict delay constraint, the redundancy decays at best polynomially with the delay constraint\footnote{This is in fact true even under the weaker expected delay constraint.}. As we shall see, the redundancy can be made to decay \textit{exponentially} with the delay, if the more general family of delay-constrained encoders is used. This reveals a fundamental difference between block/phrase length and delay in lossless source coding.

The following lemma shows that for an optimal $d$-delay-constrained encoder, the inf--redundancy and sup--redundancy coincide.
\begin{lemma}\label{lem:infsup_coincide}
For any source $P$,
\begin{equation*}
\inf_{\m{E}\in\mf{C}_d}\overline{\mf{R}}^\m{E}(P) = \inf_{\m{E}\in\mf{C}_d}\underline{\mf{R}}^\m{E}(P) \dfn \mf{R}(P,d)
\end{equation*}
\end{lemma}
\begin{proof}
See the Appendix.
\end{proof}
Accordingly, $\mf{R}(P,d)$ defined above is called the \textit{redundancy--delay function} associated with the source $P$. The corresponding inf--redundancy--delay and sup--redundancy--delay exponents associated with $P$ can now be defined:
\begin{align*}
\overline{E}(P) &=
\limsup_{d\rightarrow\infty}-\frac{1}{d}\,\log \mf{R}(P,d)
\\
\underline{E}(P) &= \liminf_{d\rightarrow\infty}-\frac{1}{d}\,\log
\mf{R}(P,d)
\end{align*}
Our main goal in this paper is to characterize $\mf{R}(P,d)$, $\overline{E}(P)$ and $\underline{E}(P)$.

\section{The Delay Profile of Arithmetic Coding}\label{sec:delay_AC}
Consider a case where a source $P$ is encoded by a mismatched arithmetic encoder, namely where the encoder's interval lengths match a different source $Q$ (see also Subsection \ref{subsec:encoders}). Note that we can always assume that $P\ll Q$, as otherwise the mismatched encoder is not well defined for all input symbols. In the next theorem we upper bound the probability that the corresponding delay profile exceeds a given threshold. This result will serve as a tool in the next section, where we lower bound the redundancy--delay exponent.
\begin{theorem}\label{thrm:memoryless}
Suppose a source $P\in\ms{P(\m{X})}$ is encoded using an arithmetic encoder $\m{E}$ matched to a source $Q\in\ms{P(\m{X})}$, where $P\ll Q$. Then
\begin{align}\label{eq:delay_prob}
\nonumber \Prob\left(\Delta^\m{E}(P) > d\right) &\leq 2\pmax^d\left(d\log{\left(\frac{\nu(P,Q)}{\pmax}\right)}+\kappa\right )
\\
&\quad+ 2q_{\rm max}^d(\nu(P,Q))^d
\end{align}
where $\kappa = \log\left(\frac{\sqrt{2}e}{\log{e}}\right) \approx 1.4139\ldots$
\end{theorem}
An outline of the proof is given in Section \ref{subsec:AC_outline}. The full proof is given in Section \ref{subsec:proof_memoryless}.
\begin{corollary}
Let $\m{E}$ be an arithmetic encoder matched to a source $Q\in\ms{P(\m{X})}$, where $P\ll Q$. For any source $P\in\ms{P(\m{X})}$, if
\begin{equation*}
\qmax\cdot\nu(P,Q) <1
\end{equation*}
then the delay profile bound (\ref{eq:delay_prob}) is exponentially decaying with $d$, hence the expected delay is finite, i.e., $\m{E}\in\mf{B}(P)$. This specifically holds for all non-deterministic $P=Q$.
\end{corollary}

\begin{corollary}
Suppose the source $P$ is encoded using the arithmetic encoder matched to the source. Then
\begin{equation*}
\Prob(\Delta^\m{E}(P)  > d) \leq 2\pmax^d\left(d\log{\left(1\slash\pmax\right)}+\kappa+1\right )
\end{equation*}
\end{corollary}

\begin{remark}
A bound on the moment-generating function for matched arithmetic coding, and a corresponding exponential bound on the delay's tail distribution, were originally observed in \cite{SavariReport, arithmetic_coding_savari}. However, these bounds depend on both $\pmin$ and $\pmax$, and can therefore be arbitrarily loose. For the tail distribution, a bound depending only on $\pmax$ was originally obtained by the authors in \cite{ShayevitzDelay2006}, where it was also shown how the proof of \cite{SavariReport,arithmetic_coding_savari} can be tweaked to remove the dependency on $\pmin$. The bound obtained here is tighter than both.
\end{remark}

\begin{remark}
The bound in Theorem \ref{thrm:memoryless} can be further tightened by observing that specific orders of the alphabet $\m{X}$ are better than others in terms of the bounding technique used here. We do not pursue this direction, since we need an order-independent bound in the sequel.
\end{remark}

\subsection{Proof Outline}\label{subsec:AC_outline}
Recall the definitions of an interval--mapping encoder and of an arithmetic encoder in particular, given in Subsection \ref{subsec:encoders}. At time $n$, the sequence $x^n$ has been encoded into $\m{I}^\m{E}(x^n)$, and the decoder is so far aware only of the interval $\mbi{\m{I}^\m{E}(x^n)}$, namely the minimal binary interval containing $\m{I}^\m{E}(x^n)$. Thus the decoder is able to decode $x^m$, where $m$ is maximal such that $\mbi{\m{I}^\m{E}(x^n)}\subseteq \m{I}^\m{E}(x^m)$. Of course, $m\leq n$ where the inequality is generally strict. After $d$ more source letters are fed to the encoder, $x^{n+d}$ is encoded into $\m{I}^\m{E}(x^{n+d})$, and the entire sequence $x^n$ can be decoded at time $n+d$ if and only if\footnote{Here we are further assuming that $Q\ll P$, see Remark \ref{rem:PQ}.}
\begin{equation}\label{eq:interval_cond}
\mbi{\m{I}^\m{E}(x^{n+d})} \subseteq \m{I}^\m{E}(x^n).
\end{equation}
Now, consider the midpoint of $\mbi{\m{I}^\m{E}(x^n)}$ which by the minimality property (see Subsection \ref{subsec:encoders}) is always contained in $\m{I}^\m{E}(x^n)$. If that midpoint is contained in $\m{I}^\m{E}(x^{n+d})$ (but not as a left edge), then condition (\ref{eq:interval_cond}) cannot be satisfied; In fact, in this case the encoder cannot yield even one further bit. This observation can be generalized to a set of points which, if contained in $\m{I}^\m{E}(x^{n+d})$, $x^n$ cannot be completely {\em decoded}. For each of these points the encoder outputs a number of bits which may enable the decoder to produce source symbols, but not enough to fully decode $x^n$. The encoding and decoding delays are therefore treated here simultaneously, rather than separately as in \cite{arithmetic_coding_savari}.

\begin{remark}\label{rem:PQ}
When $Q\not\ll P$ there are ``holes'' in the interval--mapping, namely intervals corresponding to symbols where $Q(x)>0$ but $P(x)=0$. In this case, $x^n$ can be decoded at time $n+d$ if and only if $\mbi{\m{I}^\m{E}(x^{n+d})}\cap \m{I}^\m{E}(y^{n}) = \emptyset$ for any $y^n\neq x^n$. Hence condition (\ref{eq:interval_cond}) is necessary and sufficient if $Q\ll P$, and only sufficient otherwise. This point is important to note since the case where $Q\not\ll P$ appears in the sequel.
\end{remark}

After having identified the above set of \textit{forbidden points}, we clearly need to analyze the probability of avoiding them within the next $d$ instances. Loosely speaking, for an arithmetic encoder matched to the source $P$, the maximal symbol probability $\pmax$ represents the ``crudest resolution'', or the ``lowest rate'' by which we shrink our intervals, hence intuitively dictates our ability to avoid hitting forbidden points. Indeed, the probability that the encoder avoids these points is roughly $\pmax^d$. For a mismatched encoder, we get a similar expression involving $\pmax^d,\qmax^d$ and $\nu(P,Q)$ as a measure of the mismatch between the encoder and the source.

\subsection{The Forbidden Points Notion}
We now introduce some notations and prove three lemmas, required for the proof of Theorem \ref{thrm:memoryless}. Let $I=[a,b)\subseteq [0,1)$ be some interval, and $p$ some point in that interval. We say that $p$ is {\em strictly contained} in $I$ if $p\in (a,b)$. We define the
\textit{left-adjacent} of $p$ w.r.t. $I$ to be
\begin{equation*}
\ell_I(p) \dfn \min\left \{x\in [a,p)\,:\,\exists k\in\IntF^+ ,\, x
= p-2^{-k} \right \}
\end{equation*}
and the \textit{t-left-adjacent} of $p$ w.r.t. $I$ as
\begin{equation*}
\ell_I^{(t)}(p) \dfn \overbrace{(\ell_I\circ\ell_I\circ\cdots\circ\ell_I)}^t(p)\;,\quad
\ell_I^{(0)}(p) \dfn p
\end{equation*}
Notice that $\ell_I^{(t)}(p) \rightarrow a$ monotonically with $t$. We also define the \textit{right-adjacent} of $p$ w.r.t $I$ to be
\begin{equation*}
r_I(p) \dfn \max\left \{x\in (p,b)\,:\,\exists k\in\IntF^+ ,\, x =
p+2^{-k} \right \}
\end{equation*}
and $r_I^{(t)}(p)$ as the \textit{t-right-adjacent} of $p$ w.r.t. $[a,b)$ similarly, where now $r_I^{(t)}(p) \rightarrow b$ monotonically. For any $\delta < b-a$, the \textit{adjacent $\delta$-set} of $p$ w.r.t. $I$ is defined as the set of all adjacents that are not "too close" to the edges of $I$:
\begin{eqnarray*}
& S_\delta(I,p) & \dfn \left \{ x\in[a+\delta,b-\delta)\,:\,
\exists\,t\in\IntF^+\cup \{0\}\,,\right. \\ & & \left.\qquad\qquad
x = \ell^{(t)}(p) \,\vee \,x = r^{(t)}(p) \right \}
\end{eqnarray*}
Notice that for $\delta > p-a$ this set may contain only right-adjacents, for $\delta > b-p\;$ only left-adjacents, for $\delta>\frac{b-a}{2}$ it is empty, and for $\delta=0$ it may be infinite. 
\begin{lemma}\label{lem:delta_set}
The size of $S_\delta(I,p)$ is upper bounded by
\begin{equation}\label{eq_delta_set}
|S_\delta(I,p)|\leq 1+ 2\log{\frac{|I|}{\delta}}
\end{equation}
\end{lemma}
\begin{proof}
See the Appendix. 
\end{proof}

For an interval $I$, let $m(I)$ denote the midpoint of $\mbi{I}$. Note that $m(I)\in I$, by definition of $\mbi{I}$ as the minimal binary interval containing $I$. In what follows, we will be specifically interested in the adjacent $\delta$-set of $m(I)$ w.r.t. $I$. We therefore suppress the dependence on $m(I)$ and write
\begin{equation*}
S_\delta(I) \dfn S_\delta(I,m(I))
\end{equation*}
In particular, the set $S_0(I)$ will be referred to as the \textit{forbidden points} of $I$. The forbidden points play a central role in the sequel, for the following reason:
\begin{lemma}\label{lem:fp}
Condition (\ref{eq:interval_cond}) is satisfied if and only if $\m{I}^\m{E}(x^{n+d})$ does not contain forbidden points of $\m{I}^\m{E}(x^n)$, i.e.,
\begin{equation*}
\m{I}^\m{E}(x^{n+d})\cap S_0(\m{I}^\m{E}(x^n)) = \emptyset
\end{equation*}

\end{lemma}
\begin{proof}
Write $m=m(\m{I}^\m{E}(x^n))$ for short. As already discussed, if $m$ is strictly contained in $\m{I}^\m{E}(x^{n+d})$ then (\ref{eq:interval_cond}) is not satisfied. Otherwise, assume $\m{I}^\m{E}(x^{n+d})$ lies to the left of $m$. Clearly, if $\m{I}^\m{E}(x^{n+d})\subseteq [\ell(m),m)$, then $\mbi{\m{I}^\m{E}(x^{n+d})}\subseteq [\ell(m),m)$ as well, hence (\ref{eq:interval_cond}) is satisfied. However, if $\ell(m)$ is strictly contained in $\m{I}^\m{E}(x^{n+d})$ then $\mbi{\m{I}^\m{E}(x^{n+d})}$ must be the left half of $\mbi{\m{I}^\m{E}(x^n)}$, which by minimality cannot be a subinterval of $\m{I}^\m{E}(x^n)$, hence (\ref{eq:interval_cond}) is not satisfied. The same rationale also applies to $r(m)$. The lemma follows by iterating the argument.
\end{proof}

\subsection{Proof of Theorem \ref{thrm:memoryless}}\label{subsec:proof_memoryless}

The probability that the delay $\Delta^\m{E}(x^n,P)$ is larger than $d$ is equal to (or upper bounded by, when $Q\not\ll P$, see Remark \ref{rem:PQ}) the probability that (\ref{eq:interval_cond}) is not satisfied. By Lemma \ref{lem:fp}, this in turn equals the probability that $\m{I}^\m{E}(X^{n+d})$ contains none of the forbidden points of $\m{I}^\m{E}(x^n)$. To get a handle on this latter probability, the following lemma is found useful.
\begin{lemma}\label{lem:inclusion_prob}
Suppose a source $P$ is encoded using an arithmetic encoder $\m{E}$ matched to a source $Q$, where $P\ll Q$, and let $\pmax,\qmax$ be the corresponding maximal symbol probabilities. Then for any $a\in\m{I}^\m{E}(x^n)$,
\begin{equation*}
\Prob\left(a\in \m{I}^\m{E}(X^{n+d}) | X^n = x^n\right) \leq \pmax^d
\end{equation*}
and for any interval $J\subseteq\m{I}^\m{E}(x^n)$ sharing an endpoint with $\m{I}^\m{E}(x^n)$,
\begin{align*}
\Prob(J\cap \m{I}^\m{E}(X^{n+d})&\neq \emptyset | X^n=x^n)
\\
&\leq \left(\frac{|J|}{|\m{I}^\m{E}(x^n)|}+q_{\rm max}^d\right)(\nu(P,Q))^d
\end{align*}
\end{lemma}
\begin{proof}
The set $\{\m{I}^\m{E}(x^ny^d) : y^d\in\m{X}^d\}$ is a partition of $\m{I}^\m{E}(x^n)$ into intervals, and $a$ belongs to a single interval in the partition. Therefore,
\begin{align}
\nonumber \Prob&\left(a\in \m{I}^\m{E}(X^{n+d}) | X^n = x^n\right)
\\
&\leq \max_{y^d\in\m{X}^d}\Prob(X_{n+1}^{n+d} = y^d|X^n = x^n) = \pmax^d
\end{align}
establishing the first assertion. For the second assertion, write:
\begin{align}
\nonumber \Prob&(J\cap \m{I}^\m{E}(X^{n+d})\neq \emptyset | X^n=x^n) \leq \sum_{y^d:J\cap \m{I}^\m{E}(x^ny^d)\neq \emptyset} P(y^d)
\\
\nonumber &\leq \sum_{y^d:J\cap \m{I}^\m{E}(x^ny^d)\neq \emptyset} Q(y^d)\cdot(\nu(P,Q))^d
\\
\nonumber &= (\nu(P,Q))^d\sum_{y^d:J\cap \m{I}^\m{E}(x^ny^d)\neq \emptyset} \mu^{\m{E}}_d(y^d|x^n)
\\
&\leq \left(\frac{|J|}{|\m{I}^\m{E}(x^n)|}+q_{\rm max}^d\right)(\nu(P,Q))^d
\end{align}
where we have used the fact that $\max_{y^d}\mu^{\m{E}}_d(y^d|x^n) = \qmax^d$.
\end{proof}

Write $S_\delta=S_\delta(\m{I}^\m{E}(x^n))$ for short. Note that $S_\delta\subseteq S_0$, and that $S_0\backslash S_\delta$ is contained in two intervals of length $\delta$ both sharing an edge with $\m{I}^\m{E}(x^n)$. For any $\delta>0$, the delay's tail probability is bounded as follows:
\begin{align}\label{eq_delay1}
\nonumber &\Prob(\Delta^\m{E}(x^n,P) > d)
\\
\nonumber &\quad\stackrel{({\rm a})}{\leq}
\Prob\left(\mbi{\m{I}^\m{E}(X^{n+d})} \not\subseteq \m{I}^\m{E}(x^n)
| X^n=x^n\right)
\\
\nonumber &\quad\stackrel{({\rm b})}{=} \Prob\left(S_0\cap
\m{I}^\m{E}(X^{n+d}) \neq \phi | X^n=x^n\right)
\\
\nonumber &\quad\stackrel{({\rm c})}{\leq} \Prob\left(\left
(S_0\backslash S_\delta\right)\cap \m{I}^\m{E}(X^{n+d}) \neq
\phi\,\big{|}\,X^n=x^n\right)
\\
\nonumber &\quad\qquad+\Prob\left(S_\delta\cap \m{I}^\m{E}(X^{n+d}) \neq
\phi | X^n=x^n\right)
\\
\nonumber &\quad\stackrel{({\rm d})}{\leq} 2\left(\frac{\delta}{|\m{I}^\m{E}(x^n)|} + q_{\rm max}^d\right)(\nu(P,Q))^d
\\
\nonumber &\quad\qquad +\pmax^d|S_\delta|
\\
\nonumber &\quad\stackrel{({\rm e})}{\leq} 2\left(\frac{\delta}{|\m{I}^\m{E}(x^n)|} + q_{\rm max}^d\right)(\nu(P,Q))^d
\\
&\quad\qquad+\pmax^d\left(1+2\log{\frac{|\m{I}^\m{E}(x^n)|}{\delta}}\right)
\end{align}
The transitions are justified as follows:
\begin{enumerate}[(a)]
\item Condition (\ref{eq:interval_cond}) is sufficient, see discussion in Subsection \ref{subsec:AC_outline}. In most cases this would be an equality, as condition (\ref{eq:interval_cond}) would be also necessary, see Remark \ref{rem:PQ}.

\item Lemma \ref{lem:fp}.

\item Union bound over $S_0 = S_\delta\cup \left(S_0\setminus S_\delta\right)$.

\item Lemma \ref{lem:inclusion_prob}, together with a union bound over the finite number of elements in $S_0\setminus S_\delta$.

\item Lemma \ref{lem:delta_set}.
\end{enumerate}
Taking the derivative of the right-hand-side of (\ref{eq_delay1}) w.r.t. $\delta$ we find that $\delta = \log{e}\left(\frac{\pmax}{\nu(P,Q)}\right)^d
|\m{I}^\m{E}(x^n)|$ minimizes the bound. Substituting into (\ref{eq_delay1}) and noting that the bound is independent of $x^n$, (\ref{eq:delay_prob}) is proved\footnote{Observe that (\ref{eq_delay1}) holds even if $\delta>|\m{I}^\m{E}(x^n)|$, in which case our bound becomes trivial.}.

\section{A Lower Bound for $\underline{E}(P)$}\label{sec:lower_bound}
In this section we use the delay's tail distribution mentioned in the previous section, to derive an upper bound for the redundancy--delay function, and hence a lower bound on the inf--redundancy--delay exponent, via a specific arithmetic coding scheme. We emphasize that unlike \cite{jelinek68}, the presented scheme is error free, hence there is zero probability of buffer overflow. Moreover, our figure of merit is the delay in source symbols vs. the redundancy in encoded bits per symbol.

\subsection{A Finite Delay Result}

\begin{theorem}\label{thrm:rate_loss}
The redundancy--delay function for a source $P$ is upper bounded by
\begin{equation}\label{eq:rd_func_bound}
\mf{R}(P,d) \leq
2\pmax^{d-c(\pmax)}\Big{(}(d-c(\pmax))\log{(2\slash\pmax)}+1+\kappa\Big{)}^2
\end{equation}
where
\begin{equation*}
c(x) = \left\{\begin{array}{cc} 0 & x<\frac{1}{16}
\\ 2\left\lfloor\frac{1}{\log{(2\slash x)}}\right\rfloor-1  & o.w. \end{array}\right.
\end{equation*}.
\end{theorem}

\begin{corollary}
The inf--redundancy--delay exponent for a source $P$ is lower bounded by
\begin{equation*}
\underline{E}(P) \geq \log(1\slash\pmax)
\end{equation*}
\end{corollary}

\begin{proof}
Let us first describe the high-level idea behind the proof. We extend the source's alphabet by adding two \textit{fictitious symbols}, and then encode the source using a slightly mismatched arithmetic encoder. The encoder keeps track of the decoding delay, and whenever the delay reaches $d+1$, it inserts a fictitious symbol that nullifies the delay. There are three key points: 1) There exists a mapping such that there is always at least one fictitious symbol whose interval contains no forbidden points, 2) The length assigned to the fictitious symbols can be made very small, and 3) The probability of insertion, bounded via Theorem \ref{thrm:memoryless}, is also very small.

For any interval $I=[a,b)$, let
\begin{equation*}
\varphi_I(\lambda) \dfn (1-\lambda)a + \lambda b
\end{equation*}
and define define the two disjoint subintervals
\begin{equation*}
I_L \dfn
\left(\varphi_I\left(3\slash8\right)\,,\,\varphi_I\left(1\slash
2\right)\right)\,,\; I_R\dfn
\left(\varphi_I\left(1\slash2\right)\,,\,\varphi_I\left(5\slash
8\right)\right)
\end{equation*}
The first key point is established in the following Lemma.
\begin{lemma}\label{lem:forb_free}
For any interval $I\subseteq\uico$, either $I_L\cap S_0(I) = \emptyset$ or $I_R\cap S_0(I) = \emptyset$.
\end{lemma}
\begin{proof}[Proof of Lemma \ref{lem:forb_free}]
Write $m=m(\m{I}^\m{E}(x^n))$ for short. Without loss of generality, assume that $m\leq \varphi_I(1/2)$. There are two cases:
\begin{enumerate}[(1)]
\item $m\leq \varphi_I(3/8)$: It is easily verified that the right adjacent of $m$ satisfies $r(m)>\varphi_I(1/2)$, as otherwise
    \begin{equation*}
    m+2(r(m)-m)\in I
    \end{equation*}
contradicting the maximality in the definition of the right adjacent. Therefore in this case $I_L$ contains no forbidden points of $I$.

\item $m>\varphi_I(3/8)$: By our assumption $m<\varphi_n(1/2)$, hence
\begin{equation*}
r(m) - m \geq \frac{\varphi_I(1)-\varphi_I(1\slash 2)}{2}
\end{equation*}
Rewriting, we have
\begin{equation*}
r(m) \geq m+\frac{\varphi_I(1)-\varphi_I(1\slash 2)}{2} \geq \varphi_I(5/8)
\end{equation*}
and therefore $I_R$ contains no forbidden points.
\end{enumerate}
\end{proof}

Returning to the proof of Theorem \ref{thrm:rate_loss}, define an extended alphabet $\m{X}^+ = \m{X}\cup \{x_L,x_R\}$ where $x_L,x_R$ are two \textit{fictitious symbols}. Let $P^+\in\ms{P}(\m{X}^+)$ be the corresponding extension of the source $P$ to $\m{X}^+$, assigning zero probability to the fictitious symbols. For $0<\eps<\pmax$, let $P^+_\eps\in\ms{P}(\m{X}^+)$ be a source with the following symbol probabilities:
\begin{equation*}
P^+_\eps(x) = \left\{\begin{array}{ll}(1-2\eps)P(x) & x\in\m{X} \\ \eps & x\in\{x_L,x_R\} \end{array}\right.
\end{equation*}
Clearly, $\max P^+_\eps(x) = (1-2\eps)\pmax<\frac{1}{16}$ and $\nu(P^+,P^+_\eps) = \frac{1}{1-2\eps}$. Let $<$ be any order of $\m{X}$. Assuming $P^+_\eps(x)<\frac{1}{16}$ for all $x\in\m{X}^+$, and since $|I_L|=|I_R| = |I|\slash 8$, then it is easy to see there exists a order $<^+$ of $\m{X}^+$ that preserves $<$ over $\m{X}$, such that the arithmetic encoder $\m{E}$ w.r.t. $<^+$ matched to $P^+_\eps$ has the fictitious symbols $x_L,x_R$ mapped into intervals contained in $\m{I}^\m{E}(x^n)_L$ and $\m{I}^\m{E}(x^n)_R$, respectively. If the condition on $\pmax$ is not satisfied, then we can always aggregate a few symbols into a super-symbol, so that the maximal product probability satisfies the required condition (the effect of this aggregation on the delay is treated later on). To encode the source $P^+$, let us now use the arithmetic encoder for $P^+_\eps$ above together with the following fictitious symbol insertion algorithm: The encoder keeps track of the decoding delay by emulating the decoder. Whenever this delay reaches $d+1$, the encoder finds which one of $\m{I}^\m{E}(x^n)_L$ or $\m{I}^\m{E}(x^n)_R$ contains no forbidden point as guaranteed by Lemma \ref{lem:forb_free}, and inserts the corresponding fictitious symbol $x_L$ or $x_R$ respectively, hence nullifying the decoding delay. This way, the decoding delay never exceeds $d$ and no errors are incurred.

We now bound the redundancy incurred by the encoder $\m{E}'\in\mf{C}_d$ described above. There are two different sources of redundancy. The first is due to the mismatch between $P^+$ and $P^+_\eps$, and the second is due to the coding of the inserted fictitious symbol. At each time $k>d$, the probability $w_k$ for an insertion can be bounded via Theorem \ref{thrm:memoryless}:
\begin{align}\label{eq:insertion_prob}
\nonumber w_k &= \Prob(\Delta^{\m{E}'}(X^{k-d},P) > d) \leq \Prob(\Delta^{\m{E}'}(P) > d)
\\
\nonumber &\leq 2\pmax^d\left(d\log{\left(\frac{1}{(1-2\eps)\pmax}\right)}+\kappa\right )
\\
\nonumber &\quad+ 2(1-2\eps)^d\pmax^d(1-2\eps)^{-d}
\\
& = 2\pmax^d\left(d\log{\left(\frac{1}{(1-2\eps)\pmax}\right)}+\kappa+1\right )
\end{align}
Now, let $P^{+n}$ be the $n$-product of $P^+$, and write
\begin{align*}
\mf{R}_n^{\m{E}'}(P) & = \mf{R}_n^{\m{E}'}(P^+) \stackrel{({\rm a})}{\leq} R_n^{\m{E}'}(P^+) = \frac{1}{n}D(P^{+n}\|\mu_n^{\m{E}'})
\\
& \stackrel{({\rm b})}{=} \frac{1}{n}\sum_{k=1}^n \Expt\left(D(P^+\|\mu_1^{\m{E}'}(\cdot|X^{k-1}))\right)
\\
& \stackrel{({\rm c})}{=}  D(P^+\|P^+_\eps) + \frac{1}{n}\log\frac{1}{\eps}\sum_{k=1}^n w_k
\\
& \stackrel{({\rm d})}{\leq} 2\log\left(\frac{1}{\eps}\right)\pmax^d\left(d\log{\left(\frac{1}{(1-2\eps)\pmax}\right)}+\kappa+1\right )
\\
&\quad+\log\frac{1}{1-2\eps}
\\
& \stackrel{({\rm e})}{\leq} 2\log\left(\frac{1}{\eps}\right)\pmax^d\left(2d\log{\left(\frac{2}{\pmax}\right)}+\kappa+1\right) + 4\eps
\end{align*}
The transitions are justified as follows:
\begin{enumerate}[(a)]
\item Lemma \ref{lem:redundancy_bounds}.
\item The chain rule for the divergence, and the fact that $P^{+n}$ is a product (memoryless) distribution.
\item Given $X^{k-1}$, $\mu_1^{\m{E}}$ follows $P^+_\eps$ with an extra multiplication by $\eps$ if and only if $X^{k-1}$ is such that there is an insertion. Hence the the expected divergence given $X^{k-1}$ always yields the term $D(P^+\|P^+_\eps)$, and an extra $\log{1/\eps}$ multiplied by the probability of an insertion $w_k$.
\item The bound for $w_k$ given in (\ref{eq:insertion_prob}), and $D(P^+\|P^+_\eps) = \log\frac{1}{1-2\eps}$.
\item $\log\frac{1}{1-2\eps} \leq 4\eps$ for $0<\eps<\frac{1}{16}$.
\end{enumerate}

Setting $\eps=\pmax^d$, we get:
\begin{align}\label{eq:rate_lss_bound}
\nonumber \mf{R}_n^{\m{E}'}(P) &\leq  2\pmax^d\hspace{-2pt}\left(\hspace{-1pt}d\log{\left(\frac{2}{\pmax}\right)}\hspace{-2pt}+\kappa+1\hspace{-1pt}\right)d\log{\frac{1}{\pmax}} \hspace{-1pt}+ 4\pmax^d
\\
&\leq 2\pmax^d\hspace{-2pt}\left(\hspace{-1pt}d\log{\left(\frac{2}{\pmax}\right)}\hspace{-2pt}+\kappa+1\hspace{-1pt}\right)^2
\end{align}

Finally, we address the case where $\pmax>\frac{1}{16}$. As mentioned before, we aggregate a minimal number of source symbols $k$ into a super-symbol, such that $\pmax^k < \frac{1}{16}$. This
means that $1<k<\left\lfloor\frac{4}{\log{1\slash\pmax}}\right\rfloor$. We now carry out the above procedure for the $k$-product alphabet. However, since decoding is performed $k$ symbols at a time, we set our delay threshold to be $\widetilde{d}  = \left\lfloor\frac{d+1}{k}-1\right\rfloor$. Substituting the above into (\ref{eq:rate_lss_bound}) we get
\begin{align*}
\mf{R}_n^{\m{E}'}(P) &\leq 2\pmax^{k\widetilde{d}}\left (\widetilde{d} \log(2\slash \pmax^k)+\kappa+1\right)^2
\\
&\leq
2\pmax^{d-c(\pmax)}\left((d-c(\pmax))\log{(2\slash\pmax)}+\kappa+1\right)^2
\end{align*}
\end{proof}
\begin{remark}
The scheme described above also allows the encoder to change the delay constraint \textit{on the fly}, by inserting a suitable fictitious symbol in accordance to the modified constraint. Once the decoder is made aware of this change, both encoder and decoder need to simultaneously adjust the probability of the fictitious symbols.
\end{remark}

\subsection{An Asymptotic Result}
\begin{theorem}\label{thrm:asymp_exp}
The inf--redundancy--delay exponent for a source $P$ is lower bounded by the R\'enyi entropy of order $2$ of the source, i.e., 
\begin{equation*}
\underline{E}(P) \geq H_2(P)
\end{equation*}
\end{theorem}

\begin{proof}[Proof of Theorem \ref{thrm:asymp_exp}]
We construct a unit delay encoder for the product source $P^d$ using fictitious symbols in a similar way as done in Theorem \ref{thrm:rate_loss}, with an additional random coding argument. Let $<$ be a order of $\m{X}^d$ such that all super-symbols in the same type class are adjacent (and otherwise arbitrary). Let $<_{y^d}$ be a new order which is obtained by a \textit{rotation} of the order $<$, making $y^d$ the smallest element, i.e., the unique order that respects $<$ for each of the sets $\{y^d\}\cup\{z^d: y^d < z^d\}$ and $\{z^d: z^d<y^d\}$, and where the maximal element in the latter set is the maximal element under $<_{y^d}$. Finally, let $<_{y^d}^+$ be the order of $\m{X}^{d+} \dfn \m{X}^d\cup \{x_L,x_R\}$ that respects $<_{y^d}$ over $\m{X}^d$, such that the arithmetic encoder $\m{E}$ w.r.t. $<_{y^d}^+$ matched to $P^{d+}_\eps$ has the fictitious symbols $x_L,x_R$ mapped into intervals contained in $\m{I}^\m{E}(x^n)_L$ and $\m{I}^\m{E}(x^n)_R$, respectively, and are (say) of the minimal order satisfying this.

Let us now draw an i.i.d. sequence $(Y_1^d,Y_2^d,\ldots)$ with a marginal $P^d$, independent of the source sequence. At time instance $k$ (where time is now w.r.t. the product source), we use an arithmetic encoder w.r.t. the random order $<_{Y_k^d}$, and matched to $P^d_\eps$. Denote the associated random interval--mapping encoder by $\ms{E}$. It is easy to see that for any point $a\in \m{I}^\m{E}(x^{nd})$, the probability that the interval corresponding to a type $Q$ will include $a$ is upper bounded $\pmax^d$ plus the probability of the type class $T_Q$ under $P^d$, where by Lemma \ref{lem:types} the latter is upper bounded by $2^{-dD(Q\|P)}$. By the same Lemma, the probability of any super-symbol within the type class $T_Q$ is $2^{-d(D(Q\|P)+H(Q))}$. Thus,
\begin{align}
\nonumber \Prob&\left(a\in \m{I}^\ms{E}(X^{n(d+1)}) | X^{nd} = x^{nd}\right)
\\
& \leq \sum_{Q\in\ms{P}^d(\m{X})}\left(2^{-dD(Q\|P)}+\pmax^d\right)2^{-d(D(Q\|P)+H(Q))}
\end{align}
Taking the limit as $d\rightarrow\infty$, and since there is only a polynomial number of types, we obtain
\begin{align}
\nonumber &\lim_{d\rightarrow\infty}-\frac{1}{d}\log\Prob\left(a\in \m{I}^\ms{E}(X^{n(d+1)}) | X^{nd} = x^{nd}\right)
\\
\nonumber & \geq \inf_{Q\in\ms{P}(\m{X})}\hspace{-3pt}\left\{\hspace{-2pt}D(Q\|P)+H(Q)+ \min\left(D(Q\|P),\log\frac{1}{\pmax}\right)\hspace{-2pt}\right\}
\end{align}
Let $V(Q)$ denote the function over which the infimum above is taken, and assume without loss of generality that $P$ is strictly nonzero over $\m{X}$. $V(Q)$ is continuous and the infimum is taken over a compact set, hence is attained for some $Q^*\in\ms{P}(\m{X})$. Suppose that $D(Q^*\|P) > \log{1\slash \pmax}$. Let $x\in\m{X}$ be such that $P(x) = \pmax$, and suppose there exists $y\in\m{X}$ such that $P(y)< \pmax$ and $Q^*(y) > 0$. Generate a perturbed distribution $Q^\dagger$ by increasing the probability assigned by $Q^*$ to $x$ by some $\beta>0$, and decreasing the probability assigned by $Q^*$ to $y$ by the same $\beta$, leaving the other probabilities unchanged. This implies that
\begin{equation*}
D(Q^\dagger\|P)+H(Q^\dagger) < D(Q^*\|P)+H(Q^*)\,,
\end{equation*}
since the above is equivalent (by direct calculation) to $\beta\log{\left(P(x)\slash P(y)\right)} > 0$, which holds true under the assumptions made. Now, by continuity, there exists $\beta$ small enough such that $D(Q^\dagger\|P) > \log{1\slash \pmax}$. Hence $V(Q^\dagger)<V(Q^*)$ for such $\beta$, contradicting the minimality of $Q^*$. If such $y$ does not exist, then $P(x)=\pmax$ over the entire support of $Q^*$. Therefore, $D(Q^*\|P) = \log{1\slash\pmax}-H(Q^*) \leq \log{1\slash\pmax}$, in contradiction to our assumption. We conclude that $D(Q^*\|P) \leq \log{1\slash\pmax}$. Hence,
\begin{align*}
\lim_{d\rightarrow\infty}&-\frac{1}{d}\log\Prob\left(a\in \m{I}^\ms{E}(X^{n(d+1)}) | X^{nd} = x^{nd}\right)
\\
& = \min_{Q\in\ms{P}(\m{X})}\left\{2D(Q\|P)+H(Q)\right\} = H_2(P)
\end{align*}
where Lemma \ref{lem:renyi} was invoked in the last equality. Continuing this line of argument, we can essentially replace $\pmax^d$ with $2^{-dH_2(P)}$ for $d$ large enough, throughout our proofs. Therefore, the redundancy averaged over the ensemble of random $d$-delay constrained encoders is bounded by
\begin{equation}
\Expt\left(\mf{R}^\ms{E}(P)\right) = O\left(2^{-dH_2(P)}\right)
\end{equation}
and thus there exists a deterministic encoder $\m{E}$ achieving at least that expected  performance, concluding the proof.
\end{proof}

\section{An Upper Bound for $\overline{E}(P)$}\label{sec:upper_bound}
In this section we prove an upper bound on the sup--redundancy--delay exponent, hence obtaining an asymptotic lower bound for the redundancy--delay function. This characterizes the best possible redundancy achievable by any delay-constrained encoder. Our bound holds for \textit{almost any} memoryless source, which is meant w.r.t. the Lebesgue measure over the probability simplex. 
\begin{theorem}\label{thrm:delay_redundancy}
For almost any memoryless source $P$, the sup--redundancy--delay exponent is upper bounded by
\begin{equation}\label{eq:E_LB}
\overline{E}(P) \leq 8\log\left(\frac{|\m{X}|}{\pmin}\right)
\end{equation}
\end{theorem}
\begin{remark}
Note that (\ref{eq:E_LB}) cannot hold for all sources, e.g. for 2-adic sources we can have zero redundancy with zero delay, hence an infinite exponent.
\end{remark}
\begin{remark}\label{rem:ime}
When restricted to interval--mapping encoders only, a tighter upper bound of $8\log\left(1\slash \pmin\right)$ holds.
\end{remark}

\subsection{Proof Outline}
Since the proof is somewhat tedious, we find it instructive to provide a rough outline under the assumption that the encoder admits an interval--mapping representation (rather than a generalized one). This assumption will be removed in the proof itself. Due to the strict delay constraint, at any time instance the encoder must map the next $d$ symbols into intervals that do not contain any forbidden points\footnote{As mentioned in Remark \ref{rem:PQ}, avoiding forbidden points is not always a necessary condition. However, in the next section we verify this is not a restriction.}. Typically (for almost every interval), we will find an infinite number of forbidden points concentrated near the edges, with a typical ``concentration region" whose size depends on the specific interval. Clearly, the distances between consecutive points diminishes exponentially to zero. Therefore, mapping symbols to the concentration region will result in a significant mismatch between the symbol probability and the interval length, and this phenomena incurs redundancy. This observation is made precise in Lemma
\ref{lem:redundancy_left}.

Now, loosely speaking, there are two opposing strategies the encoder may use when mapping symbols to intervals. The first is to think short-term, namely to be as faithful to the source as
possible by assigning interval lengths closely matching symbol probabilities (within the forbidden points constraint). This will likely cause the next source interval to have a relatively large
concentration region, resulting in an inevitable redundancy at the subsequent mapping. The second strategy is to think long-term, by mapping to intervals with a small concentration region. This in general cannot be done while still being faithful to the source's distribution, hence this strategy also incurs in an inevitable redundancy. The latter observation is made precise in Lemma
\ref{lem:regular2}. Our bound results from the tension between these two counterbalancing sources of redundancy.

\subsection{Proof of Theorem
\ref{thrm:delay_redundancy}}\label{subsec:proof}
In light of Lemma \ref{lem:gim_rep}, we can restrict our discussion to generalized interval--mapping encoders of the form (\ref{eq:int_map_for_d}). However, we will find it more convenient to consider a broader family of generalized interval--mapping encoders, satisfying the following conditions:
\begin{enumerate}[(i)]
\item For any $s\in\m{X}^*$, $\m{I}^\m{E}(s)$ is a union of at most $|\m{X}|^d$ intervals.\footnote{To disambiguate the statement, we clarify that any two intervals whose union is an interval are counted as a single interval.}
\item For any $s\in\m{X}^*,x^d\in\m{X}^d$, $\m{I}^{\m{E}}(sx^d)$ contains no forbidden points from any of the intervals comprising $\m{I}^\m{E}(s)$.\footnote{Note that this is satisfied by (\ref{eq:int_map_for_d}), since $\mbi{\m{I}^\m{E}(sx^d)}$ is always contained in one of the intervals comprising $\m{I}^\m{E}(s)$.}
\end{enumerate}

Let $I\subseteq\uico$ be a finite union of disjoint intervals $\{I_k\}_{k=1}^K$. Recall that $S_0(I_k)$ is the set of all forbidden points in the interval $I_k$. Define: 
\begin{equation*}
A(I) \dfn \bigcup_{k=1}^K\left\{\frac{|a-b|}{|I|}: a,b\in S_0(I_k), (a,b)\cap S_0(I_k)=\emptyset\right\}
\end{equation*}
and let
\begin{equation*}
\delta_{I} = \delta_{I}(P,d) \dfn \max \{a\in A(I) : a < \pmin^d\slash 4\}
\end{equation*}
Namely, $\delta_{I}$ is the maximal distance between two consecutive forbidden points in some $I_k$, normalized by the measure of $I$, that is smaller than $\pmin^d\slash 4$.
\begin{lemma}\label{lem:redundancy_left}
$r_d(x^n) > \delta_{\m{I}^\m{E}(x^n)}$
\end{lemma}
\begin{proof}
See the Appendix.
\end{proof}

A number $a\in\uico$ is called $(m,\ell)$--\textit{constrained} if
\begin{equation*}
a =
0.\underbrace{00\ldots0}_{m'(a)}\,\underbrace{1\phi\ldots\phi}_{m}\,\underbrace{00\ldots
0}_{\ell}\,\phi\ldots
\end{equation*}
where $m'(a)$ is the length of the zeros prefix of $a$, and $\phi$
is the ``don't care'' symbol. The $(m,\ell)$--\textit{constrained
region} $\m{C}_{m,\ell}$ is the set of all such numbers. A number $a\in\uico$ is called $(m,\ell)$--\textit{violating} if
\begin{equation}\label{eq:ml_violating}
a =
0.\underbrace{00\ldots0}_{m'(a)}\,\underbrace{1\phi\ldots\phi}_{m}\,\underbrace{\phi\ldots\ldots\ldots\ldots\phi}_{\ell\;\text{ bits, not all '0' or all '1'}}\phi\ldots
\end{equation}
The $(m,\ell)$--\textit{violating region} $\m{V}_{m,\ell}$ is the
set of all such numbers. The complement $\overline{\m{V}}_{m,\ell}
= \uico\setminus\m{V}_{m,\ell}\,$ is called the $(m,\ell)$--\textit{permissible region}.
Define the regions\footnote{The $\log$ and $\frc{\cdot}$ operations are taken pointwise on the set elements.}
\begin{equation*}
L\m{C}_{m,\ell} \dfn \frc{-\log \m{C}_{m,\ell}}\,,\quad L\m{\overline{V}}_{m,\ell}\dfn \frc{-\log
\overline{\m{V}}_{m,\ell}}
\end{equation*}
and let
\begin{equation*}
\m{D}^{(1)}_{m,\ell} \dfn
\frc{L\m{\overline{V}}_{m,\ell}-L\m{C}_{m,\ell}} \,, \quad
\m{D}^{(2)}_{m,\ell} \dfn \frc{\m{D}^{(1)}_{m,\ell} -
\m{D}^{(1)}_{m,\ell}}
\end{equation*}
The following two lemmas are easily observed.
\begin{lemma}\label{lem:approx}
Let $\mu>0$. If $a\in\m{V}_{m,\ell}$ and $b\in\m{C}_{m,\ell'}$ where $\ell<\ell'$, then
\begin{equation*}
|a-b| \geq 2^{-m'(a)}\cdot 2^{-(m+\ell)} \geq
\frac{a}{2}\cdot 2^{-(m+\ell)}
\end{equation*}
\end{lemma}
\begin{lemma}\label{lem:conv}
If $I,J\subseteq\uico$ are each a union of at most $M$ intervals of size no larger than $r$ each, then $\frc{I-J}$ can be written as a union of at most $M^2+1$ intervals of size no larger than $2r$ each.
\end{lemma}

The $(m,\ell)$--permissible region within the interval $[1/2,1)$ is comprised of $2^{m-1}+1$ subintervals. By definition, the size of each is upper-bounded by $2^{-(m'+m+\ell)+1}$. Applying
$\frc{-\log(\cdot)}$ to all such intervals in the $[1/2,1)$ interval (corresponding to $m'=0$) will stretch each of them by a factor of at most $2\log{e}<4$. All other permissible intervals (those with $m'>0$) coincide on the unit interval after applying the $\frc{-\log(\cdot)}$ operator. Hence $L\overline{\m{V}}_{m,\ell}$ can be written as a union of at most $2^{m-1}+1$ intervals, each of size at most $2^{-(m+\ell)+3}$. A similar argument shows that $L\overline{\m{V}}_{m,\ell}$ can also be written that way\footnote{It can in fact be written as a union of less and smaller intervals, but that adds nothing to our argument.}.
Appealing to Lemma \ref{lem:conv}, $\m{D}^{(1)}_{m,\ell}$ can be written as a union of at most $(2^{m-1}+1)^2+1$ intervals, each of size at most $2^{-(m+\ell)+4}$. Applying the Lemma again, we find that $\m{D}^{(2)}_{m,\ell}$ can be written as a union of at most $((2^{m-1}+1)^2+1)^2+1\leq 2^{4m+1}$ intervals each of size at most $2^{-(m+\ell)+5}$. Hence,
\begin{equation}\label{eq:conv_size}
|\m{D}^{(2)}_{m,\ell}| < 2^{4m+1} \cdot 2^{-(m+\ell)+5}  = 2^{3m-\ell+6}
\end{equation}

A source $P$ is called $(\mu_0,\lambda)$-\textit{regular} if there exists a pair of symbols $y,z\in\m{X}$ and $m_0\in\NaturalF$ such that for any $\mu\geq \mu_0$
\begin{equation}\label{eq:source_lambda}
\lambda = \left\langle\log\frac{P(y)}{P(z)}\right\rangle\not\in\bigcup_{m=m_0}^\infty\m{D}^{(2)}_{m,\lceil\mu
m\rceil}
\end{equation}
\begin{remark}
$0\in\m{D}^{(2)}_{m,\lceil\mu m\rceil}$ for any $m$ and $\mu$, hence no source can be $(\mu_0,0)$--regular. Since for a dyadic source $\lambda=0$ for any pair $y,z$, a dyadic source is never $(\mu_0,\lambda)$-regular.
\end{remark}

The following two lemmas establish some properties of $(\mu_0,\lambda)$-regularity.
\begin{lemma}\label{lem:regular}
Let $\mu_0>3$. Almost any source is $(\mu_0,\lambda)$-regular for some $\lambda>0$.
\end{lemma}
\begin{proof}
See the Appendix.
\end{proof}

Define the following set:
\begin{equation*}
A^d_{\alpha,\beta} \dfn  \left\{x^d\in\m{X}^d : \frc{-\log{P(x^d)}}\not\in \m{D}^{(1)}_{\lceil \alpha d\rceil,\lceil \beta d\rceil}\right\}
\end{equation*}
\begin{lemma}\label{lem:regular2}
Suppose $P$ is a $(\mu_0,\lambda)$-regular source. Then for any $\alpha,\beta>0$ with $\beta\slash \alpha>\mu_0$
\begin{equation*}
\liminf_{d\rightarrow\infty} P(A^d_{\alpha,\beta}) \geq \frac{1}{2}
\end{equation*}
\end{lemma}
\begin{proof}
See the Appendix.
\end{proof}

From this point forward we assume $P$ is $(\mu_0,\lambda)$-regular with $\mu_0>3$. Let $\mu<\mu'$, and define the indexed sets
\begin{align*}
B_k &\dfn  \left\{x^k\in\m{X}^k : \delta_{\m{I}^\m{E}(x^k)} > \pmin^{\mu d} \right\}
\\
C(x^k) &\dfn \left\{y^d\in\m{X}^d : \delta_{\m{I}^\m{E}(x^ky^d)} > \pmin^{\mu' d} \right\}
\end{align*}
For $x^k\in B_k$, Lemma \ref{lem:redundancy_left} implies that
\begin{equation}\label{eq:Erd1}
r_d(x^k) > \pmin^{\mu d}
\end{equation}
On the other hand, $x^k\not\in B_k$ implies that the length of each interval comprising $\m{I}^\m{E}(x^k)$ must be in $\m{C}_{\lceil d\log(1\slash \pmin)\rceil,\lceil \mu d\log(1\slash \pmin)\rceil}$. Since there are at most $|\m{X}|^d$ such intervals, it must be that
\begin{equation}\label{eq:I_size1}
|\m{I}^\m{E}(x^k)|\in\m{C}_{\lceil \alpha d\rceil,\lceil \beta d\rceil}
\end{equation}
where
\begin{equation*}
\alpha\dfn \log(1\slash \pmin) + \log|\m{X}|\,,\quad \beta\dfn \mu \log(1\slash \pmin) - \log|\m{X}|
\end{equation*}
Similarly, if $y^d\not\in C(x^k)$ then
\begin{equation}\label{eq:I_size2}
|\m{I}^\m{E}(x^ky^d)|\in\m{C}_{\lceil \alpha d\rceil,\lceil \beta' d\rceil}
\end{equation}
where
\begin{equation*}
\beta'\dfn \mu' \log(1\slash \pmin) - \log|\m{X}|
\end{equation*}
For Lemma \ref{lem:regular2} to apply, we set $\mu,\mu'$ such that $\beta\slash\alpha > \mu_0$ and $\beta'\slash \alpha >\mu_0$. This yields the constraints:
\begin{equation*}
\mu' > \mu > \mu_0+ \frac{(\mu_0+1)\log{|\m{X}|}}{\log{(1\slash\pmin)}}
\end{equation*}
In what follows, we will think of $\mu'$ as arbitrarily close to $\mu$. For any $x^k\not\in B_k$ we have:
\begin{align}\label{eq:Erd2}
\nonumber &\Expt\left(r_d(X^k)+r_d(X^{k+d}) \mid  X^k = x^k\right)
\\
\nonumber &\stackrel{(\ref{item1})}{\geq} \left(\sum_{y^d\in A^d_{\alpha,\beta} \cap \overline{C(x^k)}} \left|P(y^d)-\mu^\m{E}_d(y^d|x^k)\right|\right)^2
\\
\nonumber &\quad\qquad\qquad + \pmin^{\mu' d} P(C(x^k))
\\
\nonumber &= \left(\sum_{y^d\in A^d_{\alpha,\beta} \cap \overline{C(x^k)}} \left|\frac{P(y^d)|\m{I}^\m{E}(x^k)|-
|\m{I}^\m{E}(x^ky^d)|}{|\m{I}^\m{E}(x^k)|}\right|\right)^2
\\
\nonumber &\quad\qquad\qquad + \pmin^{\mu' d} P(C(x^k))
\\
\nonumber &\stackrel{(\ref{item2})}{\geq}
\left(\frac{1}{|\m{I}^\m{E}(x^k)|}\sum_{y^d\in A^d_{\alpha,\beta} \cap \overline{C(x^k)}}
\hspace{-12pt}\frac{P(y^d)|\m{I}^\m{E}(x^k)|}{2}\,\pmin^{\lceil\alpha d\rceil +\lceil\beta d\rceil
}\right)^2
\\
\nonumber &\quad\qquad\qquad + \pmin^{\mu' d} P(C(x^k))
\\
\nonumber &= \left(\frac{P(A^d_{\alpha,\beta} \cap \overline{C(x^k))}}{2}\right)^2\cdot\pmin^{2(\alpha+\beta)d+4}
+ \pmin^{\mu' d} P(C(x^k))
\\
\nonumber &\stackrel{(\ref{item3})}{\geq} \frac{1}{4}\left[\left(P(A^d_{\alpha,\beta} \cap \overline{C(x^k))}\right)^2 \hspace{-4pt}+ P(C(x^k))\right]\pmin^{d\max(2(\alpha+\beta),\mu')+4}
\\
\nonumber &\stackrel{(\ref{item4})}{\geq} \frac{1}{4}\left[\left(P(A^d_{\alpha,\beta}) - P(A^d_{\alpha,\beta} \cap C(x^k))\right)^2 \hspace{-2pt} + P(A^d_{\alpha,\beta} \cap C(x^k))\right]
\\
\nonumber &\quad\qquad \times\pmin^{2d(\mu+1)\log{(1\slash\pmin)}+4}
\\
\nonumber &\stackrel{(\ref{item5})}{\geq} \frac{1}{4}\left(P(A^d_{\alpha,\beta})\right)^2\cdot \pmin^{2d(\mu+1)\log{(1\slash\pmin)}+4}
\\
&= \left(\frac{1}{16}+o(1)\right)\cdot \pmin^{2d(\mu+1)\log{(1\slash\pmin)}+4}
\end{align}
The inequalities are justified as follows:
\begin{enumerate}[(a)]
\item Pinsker's inequality for the divergence \cite{cover} was used, together with Lemma \ref{lem:redundancy_left} and the nonnegativity of $r_d(\cdot)$. \label{item1}
\item (\ref{eq:I_size1}) and (\ref{eq:I_size2}) hold for all the union-of-intervals lengths in the summation. Since $\frc{-\log{P(y^d)}}\not\in \m{D}^{(1)}_{\lceil \alpha d\rceil,\lceil \beta d\rceil}$ for each $y^d$ in the summation, then appealing to Lemma \ref{lem:diff_set}, we have that $P(y^d)|\m{I}^\m{E}(x^k)|\in\m{V}_{\lceil \alpha d\rceil,\lceil \beta' d\rceil}$. The inequality now follows by virtue of Lemma \ref{lem:approx}. \label{item2}
\item \label{item3} $P(A\cap \overline{C}) = P(A)-P(A\cap C)$ and $P(C)\geq P(A\cap C)$.
\item $\mu'$ can be taken to be arbitrarily close to $\mu$.  \label{item4}
\item Lemma \ref{lem:regular2} was used to lower bound the probability of the set $A^d_{\alpha,\beta}$. \label{item5}
\end{enumerate}

Combining (\ref{eq:Erd1}) and (\ref{eq:Erd2}), we get:
\begin{align*}
&\Expt(r_d(X^k)+r_d(X^{k+d}))
\\
&\;\geq \min\left(\pmin^{\mu d}, \left(\frac{1}{16}+o(1)\right)\cdot \pmin^{2d(\mu+1)\log{(1\slash\pmin)}+4}\right)
\\
&\;= \left(\frac{1}{16}+o(1)\right)\cdot \pmin^{2d(\mu+1)\log{(1\slash\pmin)}+4}
\end{align*}
This holds for any $d$-constrained encoder $\m{E}\in\mf{C}_d$, hence and plugging into Lemma \ref{lem:redundancy_bounds} we get
\begin{align*}
\underline{\mf{R}}^{\m{E}}(P) &=
\liminf_{n\rightarrow\infty}\frac{1}{2nd}\sum_{k=1}^{n}\Expt(r_d(X^k)+r_d(X^{k+d}))
\\
&\geq \left(\frac{1}{16}+o(1)\right)\cdot \frac{1}{2d}\cdot \pmin^{2d(\mu+1)\log{(1\slash\pmin)}+4}
\end{align*}
This lower bound holds for any $\mu>\mu_0+ \frac{(\mu_0+1)\log{|\m{X}|}}{\log{(1\slash\pmin)}}$. Moreover, by Lemma \ref{lem:regular} almost any source is $(\mu_0,\lambda)$-regular for any $\mu_0>3$. Therefore, we have that for almost any source
\begin{equation*}
\underline{\mf{R}}^{\m{E}}(P)  \geq  \left(\frac{1}{16}+o(1)\right)\cdot \frac{1}{2d}\cdot \pmin^{8d\log\left(\frac{|\m{X}|}{\pmin}\right)+o(d)}
\end{equation*}
and hence
\begin{equation*}
\overline{E}(P) \leq 8\log\left(\frac{|\m{X}|}{\pmin}\right)
\end{equation*}

As mentioned in Remark \ref{rem:ime}, if the encoder is restricted to be interval--mapping then a tighter upper bound $8\log(1\slash \pmin)$ holds. In this case $\m{I}^\m{E}(\cdot)$ is a single interval rather than a union of $|\m{X}|^d$ intervals, hence the proof remains the same up to the substitution $|\m{X}|\leftrightarrow 1$.

\section{Conclusions}\label{sec:conc}
The redundancy in lossless coding of a memoryless source incurred by imposing a strict end-to-end delay constraint was analyzed, and shown to decay exponentially with the delay. The associated delay-redundancy exponent was lower bounded by the R\'enyi entropy $H_2(P)$ for any source $P$, and upper bounded by $8\log{\left(|\m{X}|\slash \pmin\right)}$ for most sources. This exponential behavior should be juxtaposed against classical results in source coding, showing a polynomial decay of the redundancy with the delay. In the classical framework, the delay is identified with the block length or the maximal phrase length, which in our framework imposes a harsh restriction: The decoder is not allowed to start reproducing source symbols in the midst of a block/phrase, and the delay is repeatedly nullified at the end of each block/phrase. This means the encoder is reset at these instances, i.e., the prefix has no effect on its future behavior. Loosely speaking, the gain of exponential versus polynomial is reaped via a tighter control over the delay process, making such reset events rare. This superior performance comes however at a possible cost: in contrast to the block/phrase-constrained setup where the encoder can clear its memory and start-over in roughly constant intervals, the more general encoders discussed in this paper need to keep track of a state. The precision required for keeping the state is however finite, and can be easily derived from Lemma \ref{lem:redundancy_left}. 

In our framework, we have isolated the impact of the delay on the redundancy by letting the transmission time $n$ go to infinity. This also makes sense complexity-wise, since the per-symbol encoding complexity is determined primarily by the delay, and not by the length of the encoded sequence. In practice however, a finite transmission time forces the encoder to terminate the codeword, which in turn incurs an additional penalty of $O(n^{-1})$ in redundancy. Setting $d=O(\log{n})$ renders this additional redundancy term commensurate with the redundancy incurred by the delay constraint. Therefore, our results imply that the delay can be made logarithmic in the block length, while maintaining the same order of redundancy. Conversely, for almost all sources this is the best possible tradeoff between block length and delay. A similar  statement in the context of universal source coding was mentioned in \cite{marcelo92}, though for a somewhat different definition of the delay.

There is still a large gap between the lower and upper bounds on the redundancy--delay exponent, where the upper bound seems particularly loose. Furthermore, it remains to be seen whether the zero-measure set of sources for which the upper bound may fail to hold, can be reduced from the set of sources that do not satisfy our intricate regularity condition, to the set of dyadic sources only, which is the smallest possible.

\appendix

\begin{proof}[Proof of Lemma \ref{lem:gim_rep}]
Let us first show that $\m{I}^{\m{E}}$ satisfies the conditions for a generalized interval--mapping encoder. $\m{I}^\m{E}(sx)\subseteq \m{I}^\m{E}(s)$ is immediate from the consistency property. Let $y,z\in\m{X}$ be distinct, and assume that $\m{I}^\m{E}(sy)\cap\m{I}^\m{E}(sz)\neq \emptyset$. Then since any two binary intervals are either disjoint or one is contained in the other, then without loss of generality there exist $x^d,\wt{x}^d$ such that $\bi{\m{E}(syx^d)}\subseteq \bi{\m{E}(sz\wt{x}^d)}$, i.e., such that $\m{E}(sz\wt{x}^d)\preceq \m{E}(syx^d)$. Since $\delta^\m{E}(\cdot,\cdot) \leq d$, it must be that $sz\preceq syx^d$ , in contradiction. This verifies the disjoint nesting property.

By the consistency property, $\m{I}^\m{E}(s)\subseteq \bi{\m{E}(s)}$. Suppose that there exists a binary interval $\bi{b}$ such that $\m{I}^\m{E}(s)\subseteq\bi{b}\subset \bi{\m{E}(s)}$. Then $\m{E}(s)\prec b \preceq \m{E}(sx^d)$ for any $x^d\in\m{X}^d$, and hence by the integrity property it must be that $b\preceq \m{E}(s)$, in contradiction. Hence $\mbi{\m{I}^\m{E}(s)} = \bi{\m{E}(s)}$ for any $s\in\m{X}^*$, verifying the minimality property.
\end{proof}

\begin{proof}[Proof of Lemma \ref{lem:redundancy_bounds}]
\begin{enumerate}[(i)]
\item 
\begin{align*}
\mf{R}^{\m{E}}_n(P)  &= \bar{L}^\m{E}_n-H(P)
\\
&= \frac{1}{n}\Expt\left(-\log\left|\mbi{\m{I}^\m{E}(X^n)}\right|\right) - H(P)
\\
&\leq \frac{1}{n}\left(\Expt\left(-\log\mu^\m{E}(X^n)\right) - H(P^n)\right)
\\
& = \frac{1}{n}\sum_{x^n\in\m{X}^n}P(x^n)\log\left(\frac{P(x^n)}{\mu^\m{E}(x^n)}\right)
\\
&= R_n^\m{E}(P)
\end{align*}

\item Consider the generalized interval mapping representation of $\m{E}$ given in Lemma \ref{lem:gim_rep}. This representation satisfies $\m{I}^\m{E}(x^{n+d}) \subseteq \m{I}^\m{E}(x^n)$. Thus similarly to the above:
\begin{align*}
\mf{R}^{\m{E}}_n(P) &= \frac{1}{n}\Expt\left(-\log\left|\mbi{\m{I}^\m{E}(X^n)}\right|\right) - H(P)
\\
&\geq \frac{1}{n}\left(\Expt\left(-\log\mu^\m{E}(X^{n+d})\right) - \frac{n}{n+d}H(P^{n+d})\right)
\\
&= \left(\frac{n+d}{n}\right)R_{n+d}^\m{E}(P) + \frac{d}{n} H(P)
\end{align*}

\item For any fixed $d\in\NaturalF$,
\begin{align*}
&\frac{1}{nd}\sum_{k=1}^{n}\Expt r_d(X^k)
\\
&\quad= -H(P) + \Expt\left(\frac{1}{nd}\sum_{k=1}^n\log\frac{\mu^\m{E}_{k}(X^k)}{\mu^\m{E}_{k+d}(X^{k+d})}\right)
\\
&\quad= -H(P) + \frac{1}{nd}\sum_{k=1}^d\Expt
\log\mu^\m{E}_{k}(X^k)
\\
&\quad\;\;\qquad\qquad - \frac{1}{nd}\sum_{k=1}^d\Expt\log\mu^\m{E}_{n+k}(X^{n+k})
\\
&\quad\leq O(n^{-1})- H(P) - \frac{1}{n}\Expt\log\mu^\m{E}_{n+d}(X^{n+d})
\\
&\quad = O(n^{-1}) + \left(\frac{n+d}{n}\right)R^\m{E}_{n+d} + \frac{d}{n}\,H(P)
\\
&\quad \leq \mf{R}^\m{E}_n + O(n^{-1})
\end{align*}
Similarly,
\begin{align*}
\frac{1}{nd}\sum_{k=1}^{n}\Expt r_d(X^k) &\geq  O(n^{-1})- H(P) - \frac{1}{n}\Expt\log\mu^\m{E}_n(X^n)
\\
&=R^\m{E}_n + O(n^{-1}) \geq \mf{R}^\m{E}_n + O(n^{-1})
\end{align*}
\end{enumerate}
\end{proof}

\begin{proof}[Proof of Lemma \ref{lem:redundancy_bounds_known}]
We only need to prove (i). An arithmetic encoder matched to the source $P$ is well known to achieve zero asymptotic redundancy \cite{arithmetic_coding_jelinek}, and a bounded expected delay \cite{gallager_lecture_notes,arithmetic_coding_savari,ShayevitzDelay2006}. Therefore \begin{equation*}
\inf_{\m{E}\in\mf{L}(P)}\overline{\mf{R}}^\m{E}(P)\leq \inf_{\m{E}\in\mf{B}(P)}\overline{\mf{R}}^\m{E}(P) \leq 0
\end{equation*}

Let $\m{E}\in\mf{L}(P)$. Define $B_d$ to be the set of all suffixes that allow decoding of any prefix with delay at most $d$, i.e.,
\begin{equation*}
B_d \dfn \{y^\infty\in\m{X}^\infty : \delta^\m{E}(s,y^\infty) \leq d\;, \forall s\in\m{X}^*\}
\end{equation*}
The lossless property implies that for any $\eps>0$ there exists $d$ large enough such that
\begin{equation}
P(B_d) \geq 1-\eps
\end{equation}
Define $\bar{B}_d$ to be the set of all prefixes in $B_d$, i.e.,
\begin{equation*}
\bar{B}_d \dfn \{z^d\in\m{X}^d : z^d\prec y^\infty \in B_d\}
\end{equation*}
Note that by the very definition of $B_d$, each prefix in $\bar{B}_d$ must appear in $B_d$ with all possible suffixes. Therefore, $P(\bar{B}_d) = P(B_d) \geq 1-\eps$ for $d$ large enough. Furthermore the lossless property also implies that for any $z^d\in\bar{B}_d$, the BV codebook $C_{z^d}:\m{X}^n\mapsto\{0,1\}^*$ defined by
\begin{equation}
C_{z^d}(x^n) \dfn \m{E}(x^nz^d)
\end{equation}
is a prefix-free lossless codebook, and hence must satisfy $\Expt|C_{z^d}(X^n)|\geq nH(P)$. Write:
\begin{align*}
\bar{L}_{n+d}^\m{E}(P)  &= \frac{1}{n+d}\sum_{z^d\in\m{X}^d}P(z^d)\sum_{x^n\in\m{X}^n}P(x^n)|\m{E}(x^nz^d)|
\\
& \geq \frac{1}{n+d}\sum_{z^d\in\bar{B}_d}P(z^d)\sum_{x^n\in\m{X}^n}P(x^n)|\m{E}(x^nz^d)|
\\
&\geq \frac{1}{n+d}\sum_{z^d\in\bar{B}_d}P(z^d)\Expt|C_{z^d}(X^n)|
\\
&\geq \frac{1}{n+d}\cdot P(\bar{B}_d)\cdot nH(P) \geq \frac{(1-\eps)n}{n+d}\,H(P)
\end{align*}
Therefore,
\begin{align*}
\underline{\mf{R}}^\m{E} &= \liminf_{n\rightarrow\infty}\mf{R}^\m{E}_{n+d}(P) \geq \lim_{n\rightarrow\infty}\left(\frac{(1-\eps)n}{n+d}-1\right)H(P)
\\
& = -\eps H(P)
\end{align*}
This holds for any $\eps>0$, hence $\underline{\mf{R}}^\m{E}\geq 0$.
\end{proof}

\begin{proof}[Proof of Lemma \ref{lem:infsup_coincide}]
Let $\m{E}\in\mf{C}_d$, and set any $\eps>0$. We show that there exists another encoder $\m{E}'\in\mf{C}_d$ such that
\begin{equation*}
\overline{\mf{R}}^{\m{E}'}(P) \leq \underline{\mf{R}}^\m{E}(P)+\eps
\end{equation*}
which immediately establishes the Lemma. The encoder $\m{E}'$ will be constructed by properly terminating $\m{E}$. Set $n$ large enough such that both
\begin{equation}\label{eq:n_min}
n > d+\min\{d,\frac{2d\underline{\mf{R}}^\m{E}(P)}{\eps}\}
\end{equation}
and
\begin{equation}\label{eq:eps4}
\mf{R}_n^\m{E}(P) \leq \underline{\mf{R}}^\m{E}(P)+\eps\slash 4
\end{equation}
For any $x^{n-d}\in\m{X}^{n-d}$, define
\begin{equation*}
y^d(x^{n-d})\dfn \argmin_{z^d\in\m{X}^d}\{|\m{E}(x^{n-d}z^d)|\}
\end{equation*}
namely, $y^d(x^{n-d})$ is the suffix that results in the minimal codelength after having encoded $x^{n-d}$. Clearly,
\begin{equation}\label{eq:shorter_codelen}
n^{-1}\Expt|\m{E}(X^{n-d}y^d(X^{n-d}))| \leq \bar{L}_n^\m{E}(P)
\end{equation}

Construct the new encoder $\m{E}'$ as follows. For any $k<n-d$, let $\m{E}'(x^k) = \m{E}(x^k)$, and let $\m{E}'(x^{n-d}) = \m{E}(x^{n-d}y^d(x^{n-d}))$. For $k>n-d$, divide $x^k$ into blocks of equal size $n-d$ (with the last one possibly shorter), apply the rule above to each separately, and let $\m{E}'(x^k)$ be the concatenation thereof. Using (\ref{eq:shorter_codelen}), we have
\begin{align*}
\mf{R}_{n-d}^{\m{E}'}(P) &= (n-d)^{-1}\Expt|\m{E}'(X^{n-d})|-H(P)
\\
&\stackrel{({\rm a})}{\leq} \frac{n}{n-d}\bar{L}_n^\m{E}(P) - H(P) \leq  \frac{n}{n-d}\mf{R}_n^{\m{E}}(P)
\\
&\stackrel{({\rm b})}{\leq} \underline{\mf{R}}^\m{E}(P) + \left(\frac{d}{n-d}\underline{\mf{R}}^\m{E}(P)+\frac{n}{n-d}\cdot \eps\slash 4\right)
\\
&\stackrel{({\rm c})}{\leq} \underline{\mf{R}}^\m{E}(P) + \eps
\end{align*}
where (a) follows from (\ref{eq:shorter_codelen}), (b) follows from (\ref{eq:eps4}), and (c) follows from the assumption (\ref{eq:n_min}). Now, from the concatenated construction we have that for any $m>n-d$
\begin{align*}
\mf{R}_m^{\m{E}'}(P) &\leq \frac{\lceil m\slash (n-d)\rceil}{m}\cdot(n-d)\cdot\mf{R}_{n-d}^{\m{E}'}(P)
\\
&\leq \frac{m+n-d}{m}\left(\underline{\mf{R}}^\m{E}(P) + \eps\right)
\end{align*}
and hence
\begin{align*}
\overline{\mf{R}}^{\m{E}'}(P) &= \limsup_{m\rightarrow\infty}\mf{R}_m^{\m{E}'}(P) \leq \underline{\mf{R}}^\m{E}(P) + \eps
\end{align*}
as desired.
\end{proof}

\begin{proof}[Proof of Lemma \ref{lem:delta_set}]
It is easy to see that the number of t-left-adjacents of $p$ that are larger than $a+\delta$ is the number of ones in the binary expansion of $(p-a)$ up to resolution $\delta$. Similarly, the number of t-right-adjacents of $p$ that are smaller than $b-\delta$ is the number of ones in the binary expansion of $(b-p)$ up to resolution $\delta$. Defining $\lceil x \rceil^+ \dfn \max (\lceil x \rceil,0 ) $, we get:
\begin{eqnarray*}
\nonumber |S_\delta(I,p)| &\leq &
\lceil\log{\frac{p-a}{\delta}}\rceil^+ +
\lceil\log{\frac{b-p}{\delta}}\rceil^+ \\ \nonumber &\leq& \left
\{\begin{array}{ll}2+\log \frac{(p-a)(b-p)}{\delta^2} &
\,,\,\delta < p-a,b-p \\ 1+ \log{\frac{|b-a|}{\delta}} & \,,\,
o.w.
\end{array}\right.
\\ &\leq& 1+ 2\log{\frac{|b-a|}{\delta}}
\end{eqnarray*}
\end{proof}

\begin{proof}[Proof of Lemma \ref{lem:redundancy_left}]
Let $I=\m{I}^\m{E}(x^n)$ throughout the proof. Let
\begin{equation*}
z^d \dfn \argmin_{y^d\in\m{Y}^d} \mu^\m{E}_d(y^d|x^n)
\end{equation*}
and let $\gamma\dfn \mu^\m{E}_d(z^d|x^n)$. If $\gamma<\delta_I$, then $z^d$ has been assigned with a measure at least four times smaller than its probability $P(z^d)$. The $d$-instantaneous redundancy can be lower bounded as follows:
\begin{align*}
\nonumber r_d(x^n) &= D(P^d\|\mu_d(\cdot|x^n)) \stackrel{({\rm a})}{\geq} D(P(z^d)\|\gamma) \stackrel{({\rm b})}{\geq} D(\pmin^d\|\gamma)
\\
\nonumber&= \pmin^d\log\frac{\pmin^d}{\gamma} + (1-\pmin^d)\log\frac{1-\pmin^d}{1-\gamma}
\\
&\stackrel{({\rm c})}{\geq} 2\pmin^d - (1-\pmin^d)\frac{\pmin^d}{1-\pmin^d} = \pmin^d \geq \delta_I
\end{align*}
In (a) we have used the data processing inequality for the divergence\footnote{Recall that $\mu_d(\cdot|x^n)$ sums to at most unity, hence can be complemented to a probability distribution by adding an auxiliary symbol $\omega$ to $\m{X}^d$ and defining $P^d(\omega) = 0$.}. In (b) we have used the fact that $\gamma<\pmin^d\leq P(z^d)$ together with the monotonicity of the scalar relative entropy. In (c) we have used $\log(1-p)\geq -\frac{p}{1-p}$ for $0<p<1$.

If on the other hand $\gamma \geq \delta_I$, then all of the $d$-fold alphabet has been assigned to a measure at most $1-\delta_I$ which results in a $d$-instantaneous redundancy lower bounded by
\begin{align*}
r_d(x^n) \geq \log\frac{1}{1-\delta_{I}} \geq \delta_I\log{e} \geq \delta_I
\end{align*}
\vspace{-0.0cm}
\end{proof}

\begin{proof}[Proof of Lemma \ref{lem:regular}]
Note that $\m{C}_{m,\ell+1}\subset \m{C}_{m,\ell}$ and $\m{V}_{m,\ell+1}\supset \m{V}_{m,\ell}$, hence $\m{D}^{(2)}_{m,\ell+} \subset \m{D}^{(2)}_{m,\ell}$.  By (\ref{eq:conv_size}), we have that for any $\mu_0>3$
\begin{align*}
\lim_{m_0\rightarrow\infty}\left|\bigcup_{\mu\geq \mu_0}\bigcup_{m=m_0}^\infty\m{D}^{(2)}_{m,\lceil\mu m\rceil}\right| &= \lim_{m_0\rightarrow\infty}\left|\bigcup_{m=m_0}^\infty\m{D}^{(2)}_{m,\lceil\mu_0 m\rceil}\right|
\\
&\leq \lim_{m_0\rightarrow\infty}\sum_{m=m_0}^\infty 2^{m(3-\mu_0)+6}
\\
&= \lim_{m_0\rightarrow\infty}\frac{2^{m_0(3-\mu_0)+6}}{1-2^{3-\mu_0}} = 0
\end{align*}
The statement of the lemma follows easily.
\end{proof}

\begin{proof}[Proof of Lemma \ref{lem:regular2}]
We will assume hereinafter that $\eps<\frac{1}{2}\pmin$. Let $y,z$ be the symbols attaining $\lambda$, and define a transformation $\sigma:\ms{P}^d(\m{X})\mapsto\ms{P}^d(\m{X})$ on types:
\begin{equation}
\sigma(Q)(x) = \left\{\begin{array}{lrcl}Q(x) & x\not\in\{y,z\} & \vee & Q(y)=0\\ Q(x)-d^{-1} & x=y & \wedge & Q(y) > 0\\ Q(x)+d^{-1} & x=z & \wedge & Q(y) > 0\end{array}\right.
\end{equation}
Namely, $\sigma$ exchanges one appearance of $y$ with the appearance of $z$ as long as this is possible, i.e., as long as $Q(y)>0$. Now, suppose $d> \frac{m_0}{\log(1\slash\pmin)}$ so that (\ref{eq:source_lambda}) is satisfied. Noting that the set $A^d_{\alpha,\beta}$ is a union of type classes, let $Q\in\ms{P}^d_\eps(\m{X},P)$ be a type such that $T_Q\cap A^d_{\alpha,\beta} = \emptyset$. Clearly $\sigma(Q)\neq Q$, and for any $x^d\in T_Q$ and $\wt{x}^d\in T_{\sigma(Q)}$,
\begin{equation*}
\frc{-\log{P(\wt{x}^d)}} = \frc{-\log{P(x^d)}+ \lambda}
\end{equation*}
Now since $\lambda\not\in\m{D}^{(2)}_{m,\lceil\mu m\rceil}$ for any $m\geq m_0$ and $\mu>\mu_0$, and since $\beta\slash \alpha >\mu_0$, then $\lambda\not\in\m{D}^{(2)}_{\lceil\alpha d\rceil,\lceil\beta d\rceil}$. Recalling the definition of $\m{D}^{(2)}_{\lceil\alpha d\rceil,\lceil\beta d\rceil}$ and appealing to Lemma \ref{lem:diff_set}, we have that $\frc{-\log{P(\wt{x}^d)}}\not\in\m{D}^{(1)}_{\lceil \alpha d\rceil,\lceil\beta d\rceil}$, hence we conclude that $\sigma(Q)\in A^d_{\alpha,\beta}$. Therefore, since $\sigma$ is one-to-one when restricted to $\ms{P}^d_\eps(\m{X},P)$, then $\sigma$ uniquely matches any type in $\ms{P}^d_\eps(\m{X},P)$ that is outside $A^d_{\alpha,\beta}$, to a type that is inside $A^d_{\alpha,\beta}$.

Let us now get a handle on the variation in the probability of a type class incurred by applying $\sigma$. It is easy to check that for any $Q\in\ms{P}^d_\eps(\m{X},P)$, and $n$ large enough,
\begin{align*}
P(T_{\sigma(Q)}) &\geq P(T_Q)\left(\frac{(P(y)-\eps)d}{(P(z)+\eps)d+1}\right)\left(\frac{P(z)}{P(y)}\right)
\\
& \geq  P(T_Q)\left(1-\frac{\eps}{P(y)}\right)\left(1-\frac{\eps+d^{-1}}{P(z)}\right)
\\
& = P(T_Q)\left(1+O(\eps) + O(d^{-1})\right)
\end{align*}
Namely, the probability of a type class for a type $Q\in\ms{P}^d_\eps(\m{X},P)$ under $P$, remains almost the same after applying $\sigma$. Therefore:
\begin{align*}
&1-P(A^d_{\alpha,\beta})
\\
&\quad\leq P\left(\bigcup_{Q\not\in\ms{P}^d_\eps(\m{X},P)}T_Q\right) + \sum_{Q\in\ms{P}^d_\eps(\m{X},P) : T_Q\cap A^d_{\alpha,\beta} = \emptyset} P(T_Q)
\\
&\quad\leq o(1) + \sum_{Q\in\ms{P}^d_\eps(\m{X},P), T_Q\cap A^d_{\alpha,\beta} = \emptyset} \frac{P(T_{\sigma(Q)})}{1+O(\eps)+O(d^{-1})}
\\
&\quad\leq o(1) + \sum_{Q : T_Q\subset A^d_{\alpha,\beta}} \frac{P(T_Q)}{1+O(\eps)+O(d^{-1})}
\\
&\quad = o(1) + \frac{P(A^d_{\alpha,\beta})}{1+O(\eps)+O(d^{-1})}
\end{align*}
Where we have used the AEP (Lemma \ref{lem:types}) in the second inequality. The result now follows by rearranging the terms above, taking the limit as $d\rightarrow\infty$, and noting that $\eps>0$ can be taken to be arbitrarily small.
\end{proof}

\section*{Acknowledgments}
We would like to thank Yuriy Reznik for pointing out Khodak's paper. We are also grateful to the anonymous reviewers for their insightful comments and suggestions that have helped improve the presentation of the paper.

\bibliographystyle{IEEEbib}
\bibliography{C:/Work/latex/ofer_refs_master}

\end{document}